%% file: main.tex
\documentclass[letterpaper,twocolumn,10pt]{article}
\usepackage{usenix2019_v3}
\usepackage{amsmath}
\usepackage{amsfonts}
\usepackage{amssymb}
\usepackage{amsthm}
\usepackage{enumitem}
\usepackage{soul}
\usepackage{textcomp}
\usepackage{url}

\usepackage{subcaption}
\usepackage{tikz}
\usepackage{graphicx}
\usepackage{float}
\usepackage[flushleft]{threeparttable}
\usepackage{verbatim}
\usepackage{multirow}
\usetikzlibrary{positioning}
\usetikzlibrary{calc, shapes, backgrounds}

\newtheorem{definition}{Definition}
\newtheorem{theorem}{Theorem}

\allowdisplaybreaks

\newcommand{\parhead}[1]{\noindent\textbf{#1}.}
\newcommand{\secp}{\kappa}

\newcommand{\sample}{\stackrel{\$}{\leftarrow}}
\newcommand{\fixed}{\text{fixed}}
\newcommand{\free}{\text{free}}
\newcommand{\lf}{lf}
\newcommand{\LF}{LF}
\newcommand{\adv}{\mathcal{A}}
\newcommand{\hyb}{\mathcal{H}}
\newcommand{\us}{\textmu{}s}
\newcommand{\uA}{\textmu{}A}
\newcommand{\uTESLA}{\textmu{}TESLA}
\DeclareMathAlphabet{\mathcal}{OMS}{cmsy}{m}{n}

\newcommand{\enc}{WKD-IBE}
\newcommand{\newprot}{JEDI}
\newcommand{\leaf}{leaf}
\newcommand{\leaves}{leaves}

\newcommand{\ent}{principal}
\newcommand{\Ents}{Principals}
\newcommand{\ents}{principals}
\newcommand{\Anent}{A principal}
\newcommand{\anent}{a principal}

\let\OLDthebibliography\thebibliography
\renewcommand\thebibliography[1]{
  \OLDthebibliography{#1}
  \setlength{\parskip}{0pt}
  \setlength{\itemsep}{2pt plus 0ex}
}

\newcommand{\appref}[1]{Appendix \ref{#1}}
\newcommand{\secref}[1]{\S{}\ref{#1}}
\newcommand{\figref}[1]{Fig.~\ref{#1}}

\newcommand{\tablebullets}[1]{\vspace{-2ex}\begin{itemize}[leftmargin=*,noitemsep,topsep=0ex]#1\end{itemize}\vspace{-3ex}}

\newif\iffull
\fulltrue

\usepackage{titlesec}
\titlespacing*{\section}{0pt}{4pt}{2pt}
\titlespacing*{\subsection}{0pt}{3pt}{2pt}
\titlespacing*{\subsubsection}{0pt}{2pt}{0pt}

\usepackage{titling}

\usepackage{caption}
\begin{document}

\date{}

\title{\Large \bf \newprot{}: Many-to-Many End-to-End Encryption and Key Delegation for IoT}

\author{
{\rm Sam Kumar, Yuncong Hu, Michael P Andersen, Raluca Ada Popa, and David E. Culler}\\
\textit{University of California, Berkeley} 
} 

\maketitle


\begin{abstract}
\input{abstract.tex}
\end{abstract}

\input{01_introduction.tex}
\input{02_model.tex}
\input{03_encryption.tex}
\input{04_signatures.tex}
\input{05_revocation.tex}
\input{06_implementation.tex}
\input{07_evaluation.tex}
\input{08_related.tex}
\input{09_conclusion.tex}


\section*{Availability}

The \newprot{} cryptography library is available at \url{https://github.com/ucbrise/jedi-pairing} and our implementation of the \newprot{} protocol for bw2 is available at \url{https://github.com/ucbrise/jedi-protocol}.


\section*{Acknowledgments}
We thank our anonymous reviewers and our shepherd William Enck for their invaluable feedback.
We would also like to thank students from the RISE Security Group and BETS Research Group for giving us feedback on early drafts of this paper.
This research was supported by Intel/NSF CPS-Security \#1505773 and \#20153754,
DoE \#DE-EE000768,
California Energy Commission \#EPC-15-057,
NSF CISE Expeditions \#CCF-1730628, NSF GRFP \#DGE-1752814, and gifts from the Sloan Foundation, Hellman Fellows Fund, Alibaba, Amazon, Ant Financial, Arm, Capital One, Ericsson, Facebook, Google, Intel, Microsoft, Scotiabank, Splunk and VMware.


\bibliographystyle{plain}
\bibliography{jedi}

\appendix
\iffull
\input{a_wkdibe.tex}
\input{b_sd.tex}
\input{c_hibe.tex}
\input{d_kpabe.tex}
\input{e_proof.tex}
\input{f_extensions.tex}
\fi

\end{document}

%% file: abstract.tex
As the Internet of Things (IoT) emerges over the next decade, developing secure communication for IoT devices is of paramount importance.
Achieving end-to-end encryption for large-scale IoT systems, like smart buildings or smart cities, is challenging because multiple \ents{} typically interact \emph{indirectly} via  intermediaries, meaning that the recipient of a message is not known in advance.
This paper proposes \newprot{} (\textbf{J}oining \textbf{E}ncryption and \textbf{D}elegation for \textbf{I}oT), a many-to-many end-to-end encryption protocol for IoT.
\newprot{} encrypts and signs messages end-to-end, while conforming to the decoupled communication model typical of IoT systems. \newprot{}'s keys support expiry and fine-grained access to data, common in IoT.
Furthermore, \newprot{} allows \ents{} to delegate their keys, restricted in expiry or scope, to other \ents{}, thereby granting access to data and managing access control in a scalable, distributed way.
Through careful protocol design and implementation, \newprot{} can run across the spectrum of IoT devices, including ultra low-power deeply embedded sensors severely constrained in CPU, memory, and energy consumption.
We apply \newprot{} to an existing IoT messaging system and demonstrate that its overhead is modest.

%% file: 01_introduction.tex
\section{Introduction}\label{sec:introduction}

As the Internet of Things (IoT) has emerged over the past decade, smart devices have become increasingly common.
This trend is only expected to continue, with tens of billions of new IoT devices deployed over the next few years~\cite{cisco2014internet}.
The IoT vision requires these devices to communicate to discover and use the resources and data provided by one another.
Yet, these devices collect privacy-sensitive information about users. A natural step to secure privacy-sensitive data is to use \emph{end-to-end encryption} to protect it during transit.

\newcommand{\myrequire}[1]{\vspace{0.5mm}\noindent$\triangleright$ \textbf{#1}}

Existing protocols for end-to-end encryption, such as SSL/TLS and TextSecure~\cite{frosch2016how}, focus on \emph{one-to-one} communication between two \ents{}: for example, Alice sends a message to Bob over an insecure channel. Such protocols, however, appear not to be a good fit for large-scale industrial IoT systems. Such IoT systems demand \textit{many-to-many} communication among \textbf{decoupled} senders and receivers, and require \textbf{decentralized delegation of access} to enforce which devices can communicate with which others.

We investigate existing IoT systems, which currently do not encrypt data end-to-end, to understand the requirements on an end-to-end encryption protocol like \newprot{}. We use \emph{smart cities} as an example application area, and data-collecting sensors in a large organization as a concrete use case. We identify three central requirements, which we treat in turn below:

\myrequire{Decoupled senders and receivers.}
IoT-scale systems could consist of thousands of \ents{}, making it infeasible for consumers of data (e.g., applications) to maintain a separate session with each producer of data (e.g., sensors). Instead, senders are typically \textbf{decoupled} from receivers. Such decoupling is common in \emph{publish-subscribe} systems for IoT, such as MQTT, AMQP, XMPP, and Solace~\cite{solace}.
In particular, many-to-many communication based on publish-subscribe is the \emph{de-facto} standard in smart buildings, used in systems like BOSS~\cite{haggerty2013boss}, VOLTTRON~\cite{volttron}, Brume~\cite{mehanovic2018brume} and bw2~\cite{andersen2017democratizing}, and adopted commercially in AllJoyn and IoTivity.
Senders publish messages by addressing them to \emph{resources} and sending them to a \emph{router}. Recipients \emph{subscribe} to a resource by asking the router to send them messages addressed to that resource.

Many systems for smart buildings/cities, like sMAP~\cite{haggerty2010smap}, SensorAct~\cite{arjunan2012sensoract}, bw2~\cite{andersen2017democratizing}, VOLTTRON~\cite{volttron}, and BAS~\cite{krioukov2012building}, organize resources as a \textbf{hierarchy}. A resource hierarchy matches the organization of IoT devices: for instance, smart cities contain buildings, which contain floors, which contain rooms, which contain sensors, which produce streams of readings. We represent each resource---a leaf in the hierarchy---as a Uniform Resource Indicator ({\bf URI}), which is like a file path. For example, a sensor that measures temperature and humidity might send its readings to the two URIs \path{buildingA/floor2/roomLHall/sensor0/temp} and \path{buildingA/floor2/roomLHall/sensor0/hum}. A user can subscribe to a URI prefix, such as \path{buildingA/floor2/roomLHall/*}, which represents a subtree of the hierarchy. He would then receive all sensor readings in room ``LHall.''

\begin{figure}[t]
    \centering
    \hbox{
        \hspace{-2ex}
        \includegraphics[width=\linewidth]{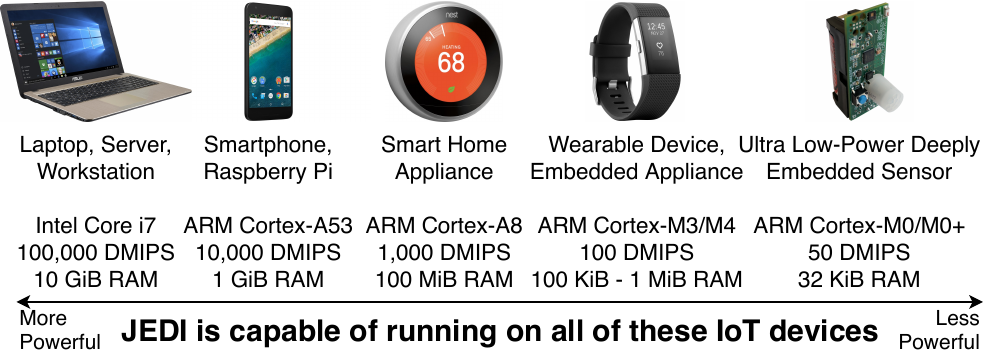}
    }
    \caption{IoT comprises a diverse set of devices, which span more than four orders of magnitude of computing power (estimated in Dhrystone MIPS).\protect\footnotemark}
    \label{fig:iot}
    \vspace{-3ex}
\end{figure}
\footnotetext{Image credits: \url{https://tweakers.net/pricewatch/1275475/asus-f540la-dm1201t.html}, \url{https://www.lg.com/uk/mobile-phones/lg-H791}, \url{https://www.bestbuy.com/site/nest-learning-thermostat-3rd-generation-stainless-steel/4346501.p?skuId=4346501}, \url{https://www.macys.com/shop/product/fitbit-charge-2-heart-rate-fitness-wristband?ID=2999458}}

\myrequire{Decentralized delegation.}
Access control in IoT needs to be fine-grained. For example, if Bob has an app that needs access to temperature readings from a single sensor, that app should receive the decryption key for only that one URI, even if Bob has keys for the entire room. In an IoT-scale system, it is not scalable for a central authority to individually give fine-grained decryption keys to each person's devices.
Moreover, as we discuss in \secref{sec:keyserver}, such an approach would pose increased security and privacy risks. Instead, Bob, who himself has access to readings for the entire room, should be able to delegate temperature-readings access to the app. Generally, \anent{} with access to a set of resources can give another \ent{} access to a subset of those resources.

Vanadium~\cite{taly2016distributed} and bw2~\cite{andersen2017democratizing} introduced \emph{decentralized delegation} (SPKI/SDSI~\cite{clarke2001certificate} and Macaroons~\cite{birgisson2014macaroons})  in the smart buildings space. Since then, decentralized delegation has become the state-of-the-art for access control in smart buildings,  especially those geared toward large-scale commercial buildings or organizations~\cite{fierro2015xbos, hviid2018activity}.
In these systems, \anent{} can access a resource if there exists a \emph{chain} of delegations, from the owner of the resource to that \ent{}, granting access. At each link in the chain, the extent of access may be qualified by \emph{caveats}, which add restrictions to which resources can be accessed and when. While these systems provide delegation of permissions, they do not provide protocols for encrypting and decrypting messages end-to-end.

\myrequire{Resource constraints.}
IoT devices vary greatly in their capabilities, as shown in \figref{fig:iot}. This includes devices constrained in CPU, memory, and energy, such as wearable devices and low-cost environmental sensors.

In smart buildings/cities, one application of interest is \emph{indoor environmental sensing}. Sensors that measure temperature, humidity, or occupancy may be deployed in a building; such sensors are \emph{battery-powered} and transmit readings using a \emph{low-power} wireless network. To see ubiquitous deployment, they must cost only \emph{tens of dollars} per unit and have \emph{several years} of battery life.
To achieve this price/power point, sensor platforms are heavily resource-constrained, with mere \emph{kilobytes} of memory (farthest right in \figref{fig:iot})~\cite{hamiltoniot, particlemesh, feldmeier2009personalized, li2014energy, brunelli2014povomon, andersen2017hamilton, andersen2016system}.
The \emph{power consumption} of encryption is a serious challenge, even more so than its latency on a slower CPU; the CPU and radio must be used sparingly to avoid consuming energy too quickly~\cite{ye2002energy, kim2018system}.
For example, on the sensor platform used in our evaluation, an average CPU utilization of merely 5\% would result in less than a year of battery life, \emph{even if the power cost of using the transducers and network were zero}.

\begin{figure}[t]
    \centering
    \includegraphics[width=\linewidth]{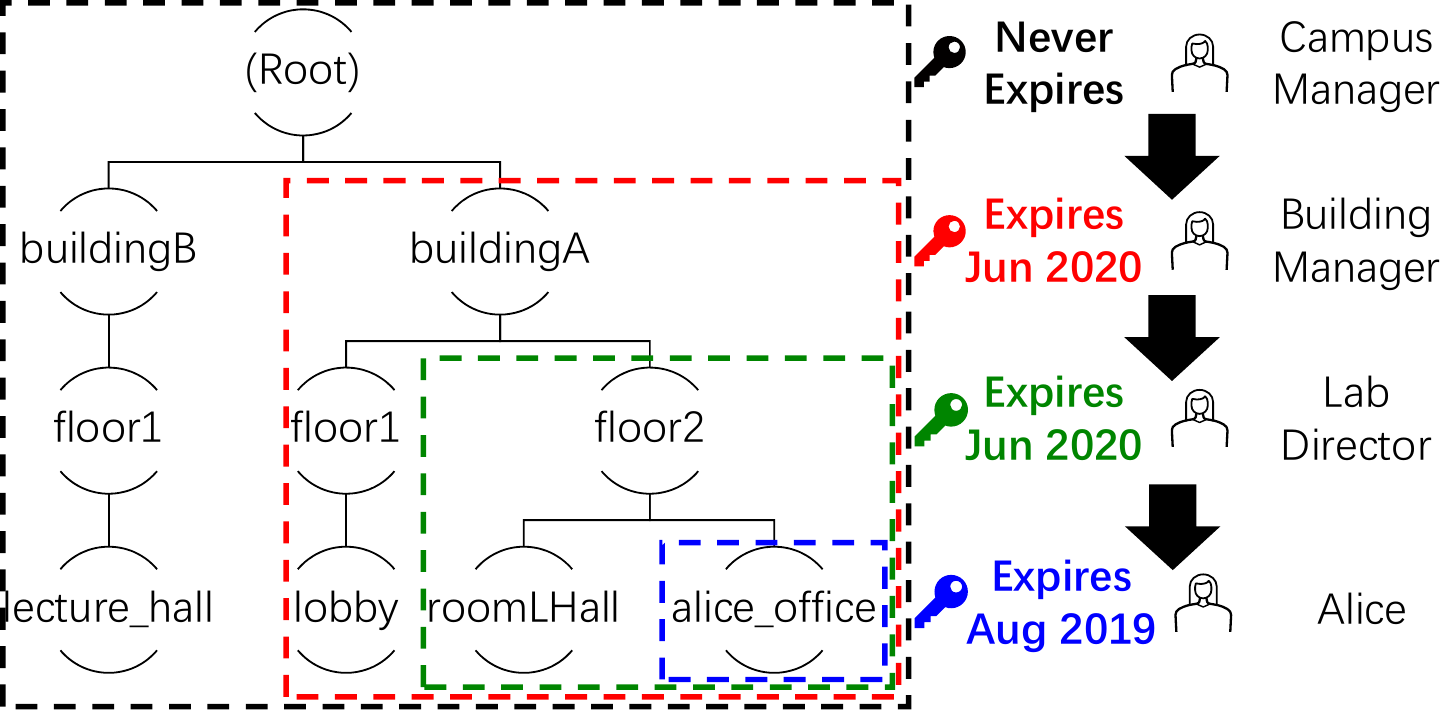}
    \caption{\newprot{} keys can be qualified and delegated, supporting decentralized, cryptographically-enforced access control via key delegation. Each person has a decryption key for the indicated resource subtree that is valid until the indicated expiry time. Black arrows denote delegation.}
    \label{fig:decentralized}
    \vspace{-4mm}
\end{figure}

\subsection{Overview of \newprot{}}\label{s:overview}\label{ssec:challenges}

This paper presents  \newprot{}, a \emph{many-to-many} end-to-end encryption protocol compatible with the above three requirements of  IoT systems. \newprot{} encrypts messages end-to-end for confidentiality,  signs them for integrity while preserving anonymity, and supports delegation with caveats, all while allowing senders and receivers to be decoupled via a resource hierarchy.
\newprot{} differs from existing encryption protocols like SSL/TLS, requiring us to overcome a number of \emph{challenges}:
\begin{enumerate}[leftmargin=*,topsep=0mm,noitemsep]
    \item Formulating a new system model for end-to-end encryption to support \textbf{decoupled senders and receivers} and \textbf{decentralized delegation} typical of IoT systems (\secref{s:intro_model})
    \item Realizing this expressive model while working within the \textbf{resource constraints} of IoT devices (\secref{s:intro_encryption})
    \item Allowing receivers to verify the integrity of messages, while preserving the anonymity of senders (\secref{s:intro_integrity})
    \item Extending \newprot{}'s model to support revocation (\secref{s:intro_revocation})
\end{enumerate}
Below, we explain how we address each of these challenges.

\subsubsection{\newprot{}'s System Model (\secref{sec:security})}\label{s:intro_model}

Participants in \newprot{} are called \emph{\ents{}}. Any \ent{} can create a \textbf{resource hierarchy} to represent some resources that it owns.
Because that \ent{} owns all of the resources in the hierarchy, it is called the \emph{authority} of that hierarchy.

Due to the setting of \textbf{decoupled senders and receivers}, the sender can no longer encrypt messages with the receiver's public key, as in traditional end-to-end encryption. Instead, \newprot{} models \ents{} as interacting with resources, rather than with other \ents{}. Herein lies the key difference between \newprot{}'s model and other end-to-end encryption protocols: \ul{the publisher of a message encrypts it according to the URI to which it is published, not the recipients subscribed to that URI.}
Only \ents{} permitted to subscribe to a URI are given keys that can decrypt messages published to that URI.

IoT systems that support \textbf{decentralized delegation} (Vanadium, bw2), as well as related non-IoT authorization systems (e.g., SPKI/SDSI~\cite{clarke2001certificate} and Macaroons~\cite{birgisson2014macaroons}) provide \ents{} with tokens (e.g., certificate chains) that they can present to prove they have access to a certain resource.
Providing tokens, however, is not enough for end-to-end encryption; unlike these systems, \newprot{} allows \emph{decryption keys} to be distributed via chains of delegations. Furthermore, the URI prefix and expiry time associated with each \newprot{} key can be restricted at each delegation.
For example, as shown in \figref{fig:decentralized}, suppose Alice, who works in a research lab, needs access to sensor readings in her office. In the past, the campus facilities manager, who is the authority for the hierarchy, granted a key for \path{buildingA/*} to the building manager, who granted a key for \path{buildingA/floor2/*} to the lab director. Now, Alice can obtain the key for \path{buildingA/floor2/alice_office/*} directly from her local authority (the lab director).

\subsubsection{Encryption with URIs and Expiry (\secref{sec:encryption})}\label{s:intro_encryption}

\newprot{} supports \emph{decoupled} communication. The resource to which a message is published acts as a \emph{rendezvous point} between the senders and receivers, used by the underlying system to route messages. Central to \newprot{} is the challenge of finding an analogous \emph{cryptographic rendezvous point} that senders can use to encrypt messages without knowledge of receivers. A number of IoT systems~\cite{shafagh2018droplet,perrig2001spins} use only simple cryptography like AES, SHA2, and ECDSA, but these primitives are not expressive enough to encode \newprot{}'s rendezvous point, which must support hierarchically-structured resources, non-interactive expiry, and decentralized delegation.

Existing systems \cite{wang2016sieve, wang2010hierarchical, wang2011hierarchical} with similar expressivity to \newprot{} use Attribute-Based Encryption (ABE)~\cite{goyal2006attribute, bethencourt2007ciphertext}. Unfortunately, ABE is not suitable for \newprot{} because it is too expensive, especially in the context of \textbf{resource constraints} of IoT devices. Some IoT systems rule it out due to its latency alone~\cite{shafagh2018droplet}.
In the context of low-power devices, encryption with ABE would also consume too much power.
\newprot{} circumvents the problem of using  ABE or basic cryptography with two insights: (1) Even though ABE is too heavy for low-power devices, this does not mean that we must resort to only symmetric-key techniques. We show that certain IBE schemes~\cite{abdalla2007generalized} can be made practical for such devices.
(2) \textbf{Time is another resource hierarchy}: a timestamp can be expressed as \path{year/month/day/hour}, and in this hierarchical representation, any time range can be represented efficiently as a logarithmic number of subtrees. With this insight, we can simultaneously support URIs and expiry via a nonstandard use of a certain type of IBE scheme: WKD-IBE~\cite{abdalla2007generalized}. Like ABE, WKD-IBE is based on bilinear groups (pairings), but it is an order-of-magnitude less expensive than ABE as used in \newprot{}.
To make \newprot{} practical on low-power devices, we design it to invoke WKD-IBE \emph{rarely}, while relying on AES most of the time, much like session keys. Thus, \newprot{} achieves expressivity commensurate to IoT systems that do not encrypt data---significantly more expressive than AES-only solutions---while allowing several years of battery life for low-power low-cost IoT devices.

\subsubsection{Integrity and Anonymity (\secref{sec:integrity})}\label{s:intro_integrity}
In addition to being encrypted, messages should be signed so that the recipient of a message can be sure it was not sent by an attacker.
This can be achieved via a certificate chain, as in SPKI/SDSI or bw2. Certificates can be distributed in a decentralized manner, just like encryption keys in \figref{fig:decentralized}.

Certificate chains, however, are insufficient if anonymity is required. For example, consider an office space with an occupancy sensor in each office, each publishing to the same URI \path{buildingA/occupancy}. In aggregate, the occupancy sensors could be useful to inform, e.g., heating/cooling in the building, but individually, the readings for each room could be considered privacy-sensitive. The occupancy sensors in different rooms could use different certificate chains, if they were authorized/installed by different people. This could be used to deanonymize occupancy readings.
To address this challenge, we adapt the WKD-IBE scheme that we use for end-to-end encryption to achieve an \textit{anonymous} signature scheme that can encode the URI and expiry and support decentralized delegation.
Using this technique, anonymous signatures are practical even on low-power embedded IoT devices.

\subsubsection{Revocation (\secref{sec:revocation})}\label{s:intro_revocation}

As stated above, \newprot{} keys support expiry.
Therefore, it is possible to achieve a lightweight revocation scheme by delegating each key with short expiry and periodically renewing it to extend the expiry.
To revoke a key, one simply does not renew it. We expect this expiry-based revocation to be sufficient for most use cases, especially for low-power devices, which typically just ``sense and send.''

Enforcing revocation cryptographically, without relying on expiration, is challenging. As we discuss in \secref{sec:revocation}, any cryptographically-enforced scheme that provides immediate revocation (i.e., keys can be revoked without waiting for them to expire) is subject to the fundamental limitation that the sender of a message must know which recipients are revoked when it encrypts the message.  \newprot{} provides a form of immediate revocation, subject to this constraint. We use techniques from tree-based broadcast encryption~\cite{naor2001revocation, dodis2002public} to encrypt in such a way that all decryption keys for that URI, \emph{except for ones on a revocation list}, can be used to decrypt. Achieving this is nontrivial because we have to combine broadcast encryption with \newprot{}'s semantics of hierarchical resources, expiry, and delegation. First, we modify broadcast encryption to support delegation, in such a way that if a key is revoked, all delegations made with that key are also implicitly revoked. Then, we integrate broadcast revocation, in a \emph{non-black-box} way, with \newprot{}'s encryption and delegation, as a third resource hierarchy alongside URIs and expiry.

\subsection{Summary of Evaluation}

For our evaluation, we use \newprot{} to encrypt messages transmitted over bw2~\cite{andersen2017democratizing, bw2}, a deployed open-source messaging system for smart buildings, and demonstrate that \newprot{}'s overhead is small in the critical path. We also evaluate \newprot{} for a commercially available sensor platform called ``Hamilton''~\cite{hamiltoniot}, and show that a Hamilton-based sensor sending one sensor reading every 30 seconds would see several years of battery lifetime when sending sensor readings encrypted with \newprot{}. As Hamilton is among the least powerful platforms that will participate in IoT (farthest to the right in \figref{fig:iot}), this validates that \newprot{} is practical across the IoT spectrum.

%% file: 02_model.tex
\section{\newprot's Model and Threat Model}\label{sec:security}\label{sec:keyserver}

\Anent{} can post a message to a resource in a hierarchy by encrypting it according to the resource's URI, hierarchy's public parameters, and current time, and passing it to the underlying system that delivers it to the relevant subscribers. Given the secret key for a resource subtree and time range, \anent{} can generate a secret key for a subset of those resources and subrange of that time range, and give it to another \ent{}, as in \figref{fig:decentralized}. The receiving \ent{} can use the delegated key to decrypt messages that are posted to a resource in that subset at a time during that subrange.

\textit{\newprot{} does not require the structure of the resource hierarchy to be fixed in advance}. In \figref{fig:decentralized}, the campus facilities manager, when granting access to \path{buildingA/*} to the building manager, need not be concerned with the structure of the subtree rooted at \path{buildingA}. This allows the building manager to organize \path{buildingA/*} independently.

\subsection{Trust Assumptions}
\Anent{} is trusted for the resources it owns or was given access to (for the time ranges for which it was given access). In other words, an adversary who compromises \anent{} can read all resources that \ent{} can read and forge new messages as if it were that \ent{}. In particular, an adversary who compromises the authority for a resource hierarchy gains control over that resource hierarchy.

\newprot{} allows each \ent{} to act as an authority for its own resource hierarchy in its own trust domain, without a single authority spanning all hierarchies. In particular, \emph{\ents{}} are not organized hierarchically; \anent{} may be delegated multiple keys, each belonging to a different resource hierarchy. In the example in \figref{fig:decentralized}, Alice might also receive \newprot{} keys from her landlord granting access to resources in her apartment building, in a separate hierarchy where her landlord is the authority. If Alice owns resources she would like to delegate to others, she can set up her own hierarchy to represent those resources. Existing IoT systems with decentralized delegation, like bw2 and Vanadium, use a similar model.

\begin{figure}[t!]
    \centering
    \includegraphics[width=\linewidth]{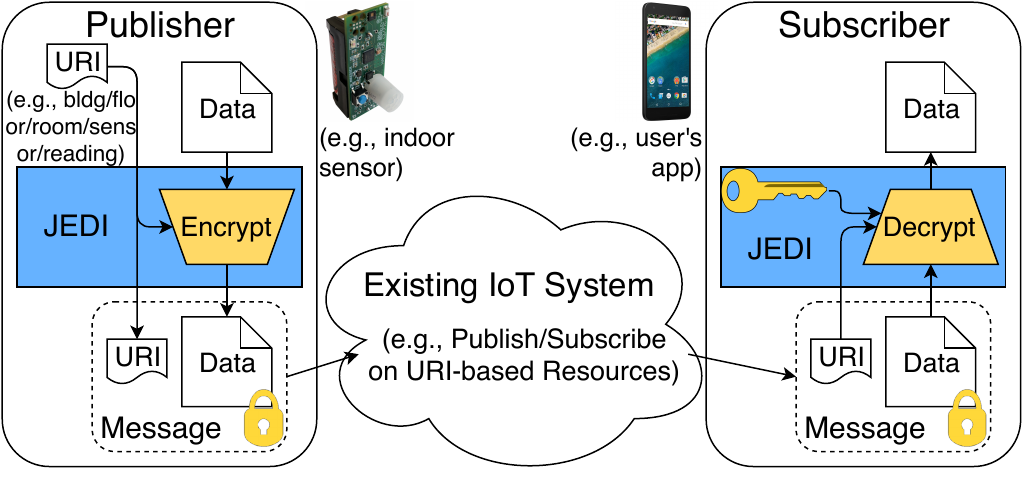}
    \caption{Applying \newprot{} to a smart buildings IoT system. Components introduced by \newprot{} are shaded. The subscriber's key is obtained via \newprot{}'s decentralized delegation (\figref{fig:decentralized}).}
    \vspace{-3ex}
    \label{fig:system}
\end{figure}

\subsection{Applying \newprot{} to an Existing System}
As shown in \figref{fig:system}, \newprot{} can be applied as a wrapper around existing many-to-many communication systems, including publish-subscribe systems for smart cities. The transfer of messages from producers to consumers is handled by the existing system. A common design used by such systems is to have a central broker (or router) forward messages; however, an adversary who compromises the broker can read all messages. In this context, \newprot{}'s end-to-end encryption protects data from such an adversary. Publishers encrypt their messages with \newprot{} before passing them to the underlying communication system (without knowledge of who the subscribers are), and subscribers decrypt them with \newprot{} after receiving them from the underlying communication system (without knowledge of who the publishers are).

\subsection{Comparison to a Na\"ive Key Server Model}
To better understand the benefits of \newprot{}'s model, consider the natural strawman of a trusted key server. This key server generates a key for every URI and time. A publisher encrypts each message for that URI with the same key. A subscriber requests this key from the trusted key server, which must first check if the subscriber is authorized to receive it.
The subscriber can decrypt messages for a URI using this key, and contact the key server for a new key when the key expires.
\newprot{}'s model is better than this key server model as follows:
\begin{itemize}[itemsep=0mm,leftmargin=*,topsep=0ex,noitemsep]
    \item \textit{Improved security.} Unlike the trusted key server, which must always be online, the authority in \newprot{} can delegate qualified keys to some \ents{} \emph{and then go offline}, leaving these \ents{} to qualify and delegate keys further. While the authority is offline, it is more difficult for an attacker to compromise it and easier for the authority to protect its secrets because it need only access them rarely.
    This reasoning is the basis of root Certificate Authorities (CAs), which access their master keys infrequently.
    In contrast, the trusted key server model requires a central trusted party (key server) to be online to grant/revoke access to any resource.

    \item \textit{Improved privacy.} No single participant sees all delegations in \newprot{}. An adversary in \newprot{} who steals an authority's secret key can decrypt all messages for that hierarchy, but still does not learn who has access to which resource, and cannot access separate hierarchies to which the first authority has no access.
    In contrast, an adversary who compromises the key server learns who has access to which resource and can decrypt messages for all hierarchies.

    \item \textit{Improved scalability.} In the campus IoT example above, if a building admin receives access to all sensors and all their different readings for a building, the admin must obtain a potentially very large number of keys, instead of one key for the entire building. Moreover, the campus-wide key server needs to grant decryption keys to each application owned by each employee or student at the university.
    Finally, the campus-wide key server must understand which delegations are allowed at lower levels in the hierarchy, requiring the entire hierarchy to be centrally administered.

\end{itemize}

\subsection{IoT Gateways}
Low-power wireless embedded sensors, due to power constraints, often do not use network protocols like Wi-Fi, and instead use specialized low-power protocols such as Bluetooth or IEEE 802.15.4. It is common for these devices to rely on an \emph{application-layer gateway} to send data to computers outside of the low-power network~\cite{zachariah2015internet}. This gateway could be in the form of a phone app (e.g., Fitbit), or in the form of a specialized border router~\cite{zigbeegateway, winter2012rpl}. In some traditional setups, the gateway is responsible for performing encryption/authentication~\cite{perrig2001spins}. \newprot{} accepts that gateways may be necessary for Internet connectivity, but does not rely on them for security---\newprot{}'s cryptography is lightweight enough to run directly on the low-power sensor nodes. This approach prevents the gateway from becoming a single point of attack; an attacker who compromises the gateway cannot see or forge data for any device using that gateway.

\subsection{Generalizability of \newprot{}'s Model}
Since \newprot{} decouples senders from receivers, it has no requirements on what happens at any intermediaries (e.g., does not require messages to be forwarded from publishers to subscribers in any particular way).
Thus, \newprot{} works even when messages are exchanged in a broadcast medium, e.g., multicast.
This also means that \newprot{} is more broadly applicable to systems with hierarchically organized resources. For example, URIs could correspond to filepaths in a file system, or URLs in a RESTful web service.

\subsection{Security Goals}
\newprot{}'s goal is to ensure that \ents{} can only read messages from or send messages to resources they have been granted access to receive from or send to.
In the context of publish-subscribe, \newprot{} also hides the content of messages from an adversary who controls the router.

\newprot{} does not attempt to hide metadata relating to the actual transfer of messages (e.g., the URIs on which messages are published, which \ents{} are publishing or subscribing to which resources, and timing).
Hiding this metadata is a complementary task to achieving delegation and end-to-end encryption in \newprot{}, and techniques from the secure messaging literature~\cite{cheng2016talek, corrigan2015riposte, van2015vuvuzela} will likely be applicable.

%% file: 03_encryption.tex
\section{End-to-End Encryption in \newprot{}}\label{sec:encryption}

A central question answered in this section is: How should publishers encrypt messages before passing them to the underlying system for delivery (\secref{ssec:uri_time})?
As explained in \secref{s:intro_encryption}, although ABE, the obvious choice, is too heavy for low-power devices, we identify \enc{}, a more lightweight identity-based encryption scheme, as sufficient to achieve \newprot{}'s properties. The primary challenge is to encode a sufficiently expressive rendezvous point in the \enc{} ID (called a \emph{pattern}) that publishers use to encrypt messages (\secref{ssec:uri_time}).

\subsection{Building Block: \enc{}}\label{ssec:enc}

We first explain \enc{}~\cite{abdalla2007generalized}, the encryption scheme that \newprot{} uses as a building block. Throughout this paper, we denote the security parameter as $\secp$.

In \enc{}, messages are encrypted with \emph{patterns}, and keys also correspond to patterns. A pattern is a list of values:  $P = (\mathbb{Z}_p^* \cup \{ \bot \})^\ell$. The notation $P(i)$ denotes the $i$th component of $P$, 1-indexed. A pattern $P_1$ \emph{matches} a pattern $P_2$ if, for all $i \in [1, \ell]$, either $P_1(i) = \bot$ or $P_1(i) = P_2(i)$. In other words, if $P_1$ specifies a value for an index $i$, $P_2$ must match it at $i$. Note that the ``matches'' operation is not commutative; ``$P_1$ matches $P_2$'' does not imply ``$P_2$ matches $P_1$''.

We refer to a component of a pattern containing an element of $\mathbb{Z}_p^*$ as \emph{fixed}, and to a component that contains $\bot$ as \emph{free}. To aid our presentation, we define the following sets:
\vspace{-1ex}
\begin{definition}
For a pattern $S$, we define:
\vspace{-1ex}
\begin{align*}
\fixed(S) &= \{ (i, S(i)) \mid S(i) \neq \bot \}\\
\free(S) &= \{ i \mid S(i) = \bot \}
\end{align*}
\end{definition}
\vspace{-1ex}
A key for pattern $P_1$ can decrypt a message encrypted with pattern $P_2$ if $P_1 = P_2$. Furthermore, a key for pattern $P_1$ can be used to derive a key for pattern $P_2$, as long as $P_1$ matches $P_2$. In summary, the following is the syntax for \enc{}.
\begin{itemize}[itemsep=0mm,leftmargin=*,topsep=0.4ex]

\item
$\mathbf{Setup}(1^\secp, 1^\ell) \rightarrow \mathsf{Params}, \mathsf{MasterKey}$;

\item
$\mathbf{KeyDer}(\mathsf{Params}, \mathsf{Key}_\mathsf{Pattern_A}, \mathsf{Pattern_B}) \rightarrow \mathsf{Key}_{\mathsf{Pattern_B}}$, derives a key for $\mathsf{Pattern_B}$,  where either $\mathsf{Key}_\mathsf{Pattern_A}$ is the $\mathsf{MasterKey}$, or $\mathsf{Pattern_A}$ matches $\mathsf{Pattern_B}$;

\item
$\mathbf{Encrypt}(\mathsf{Params}, \mathsf{Pattern}, m) \rightarrow \mathsf{Ciphertext}_{\mathsf{Pattern}, m}$;

\item
$\mathbf{Decrypt}(\mathsf{Key}_{\mathsf{Pattern}}, \mathsf{Ciphertext}_{\mathsf{Pattern}, m}) \rightarrow m$.
\end{itemize}

We use the \enc{} construction in \S{}3.2 of \cite{abdalla2007generalized}, based on BBG HIBE~\cite{boneh2005hierarchical}. Like the BBG construction, it has constant-size ciphertexts, but requires the maximum pattern length $\ell$ to be known at Setup time.
In this \enc{} construction, patterns containing $\bot$ can only be used in $\mathbf{KeyDer}$, not in $\mathbf{Encrypt}$; we extend it to support encryption with patterns containing $\bot$. We include the \enc{} construction with our optimizations in
\iffull
Appendix \ref{app:enc}.
\fi

\subsection{Concurrent Hierarchies in \newprot{}}\label{ssec:concurrent_hierarchies}
\enc{} was originally designed to allow delegation in a {\em single} hierarchy.
For example, the original suggested use case of \enc{} was to generate secret keys for a user's email addresses in all valid subdomains, such as \texttt{sysadmin@*.univ.edu}~\cite{abdalla2007generalized}.

\newprot{}, however, uses \enc{} in a nonstandard way to simultaneously support {\em multiple} hierarchies, one for URIs and one for expiry (and later in \secref{sec:revocation}, one for revocation), each in the vein of HIBE.
We think of the $\ell$ components of a \enc{} pattern as ``slots'' that are initially empty, and are progressively filled in with calls to $\mathbf{KeyDer}$.
To combine a hierarchy of maximum depth $\ell_1$ (e.g., the URI hierarchy) and a hierarchy of maximum depth $\ell_2$ (e.g., the expiry hierarchy), one can $\mathbf{Setup}$ \enc{} with the number of slots equal to $\ell = \ell_1 + \ell_2$. The first $\ell_1$ slots are filled in left-to-right for the first hierarchy and the remaining $\ell_2$ slots are filled in left-to-right for the second hierarchy (\figref{fig:slots}).

\subsection{Overview of Encryption in \newprot{}}\label{ssec:prot}

Each \ent{} maintains a \textbf{key store} containing \enc{} decryption keys.
To create a resource hierarchy, any \ent{} can call the \enc{} $\mathbf{Setup}$ function to create a resource hierarchy. It releases the \emph{public parameters} and stores the \emph{master secret key} in its key store, making it the authority of that hierarchy. To \ul{delegate} access to a URI prefix for a time range, \anent{} (possibly the authority) searches its key store for a set of keys for a superset of those permissions. It then qualifies those keys using $\mathbf{KeyDer}$ to restrict them to the specific URI prefix and time range (\secref{ssec:key_set}), and sends the resulting keys to the recipient of the delegation.\footnote{\newprot{} does not govern \emph{how} the key set is transferred to the recipient, as there are existing solutions for this. One can use an existing protocol for one-to-one communication (e.g., TLS) to securely transfer the key set. Or, one can encrypt the key set with the recipient's (normal, non-\enc{}) public key, and place it in a common storage area.} The recipient \ul{accepts the delegation} by adding the keys to its key store.

Before sending a message to a URI, \anent{} \ul{encrypts} the message using \enc{}. The pattern used to encrypt it is derived from the URI and the current time (\secref{ssec:uri_time}), which are included along with the ciphertext. When \anent{} receives a message, it searches its key store, using the URI and time included with the ciphertext, for a key to \ul{decrypt} it.

In summary, \newprot{} provides the following API:

$\mathbf{Encrypt}(\mathsf{Message}, \mathsf{URI}, \mathsf{Time}) \rightarrow \mathsf{Ciphertext}$

$\mathbf{Decrypt}(\mathsf{Ciphertext}, \mathsf{URI}, \mathsf{Time}, \mathsf{KeyStore}) \rightarrow \mathsf{Message}$

$\mathbf{Delegate}(\mathsf{KeyStore}, \mathsf{URIPrefix}, \mathsf{TimeRange}) \rightarrow \mathsf{KeySet}$

$\mathbf{AcceptDelegation}(\mathsf{KeyStore}, \mathsf{KeySet}) \rightarrow \mathsf{KeyStore}'$

Note that the \enc{} public parameters are an implicit argument to each of these functions.
Finally, although the above API lists the arguments to $\mathbf{Delegate}$ as $\mathsf{URIPrefix}$ and $\mathsf{Time Range}$, \newprot{} actually supports succinct delegation over more complex sets of URIs and timestamps (see \secref{ssec:extensions}).

\subsection{Expressing URI/Time as a Pattern}\label{ssec:uri_time}

A message is encrypted using a pattern derived from (1) the URI to which the message is addressed, and (2) the current time.
Let $H: \{0, 1\}^* \rightarrow \mathbb{Z}_p^*$ be a collision-resistant hash function. Let $\ell = \ell_1 + \ell_2$ be the pattern length in the hierarchy's \enc{} system. We use the first $\ell_1$ slots to encode the URI, and the last $\ell_2$ slots to encode the time.

Given a URI of length $d$, such as \path{a/b/c} ($d=3$ in this example), we split it up into individual components, and append a special terminator symbol \texttt{\$}: \texttt{("a", "b", "c", \$)}. Using $H$, we map each component to $\mathbb{Z}_p^*$, and then put these values into the first $d+1$ slots.
If $S$ is our pattern, we would have $S(1) = H(\texttt{"a"})$, $S(2) = H(\texttt{"b"})$, $S(3) = H(\texttt{"c"})$, and $S(4) = H(\texttt{\$})$ for this example.
Now, we encode the time range into the remaining $\ell_2$ slots. Any timestamp, with the granularity of an hour, can be represented hierarchically as \texttt{(year, month, day, hour)}. We encode this into the pattern like the URI: we hash each component, and assign them to consecutive slots. The final $\ell_2$ slots encode the time, so the depth of the time hierarchy is $\ell_2$.
The terminator symbol \texttt{\$} is not needed to encode the time, because timestamps always have exactly $\ell_2$ components.
For example, suppose that \anent{} sends a message to \path{a/b} on June 8, 2017 at 6 AM.
The message is encrypted with the pattern in \figref{fig:slots}.

\begin{figure}
    \centering
    \includegraphics[width=\linewidth]{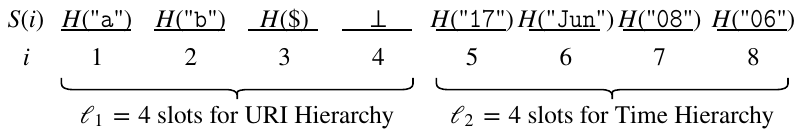}
    \caption{Pattern $S$ used to encrypt message sent to \texttt{a/b} on June 08, 2017 at 6 AM. The figure uses 8 slots for space reasons; \newprot{} is meant to be used with more slots (e.g., 20).}
    \label{fig:slots}
    \vspace{-3ex}
\end{figure}

\subsection{Producing a Key Set for Delegation}\label{ssec:key_set}

Now, we explain how to produce a key set corresponding to a URI prefix and time range. To express a URI prefix as a pattern, we do the same thing as we did for URIs, without the terminator symbol \texttt{\$}. For example, \path{a/b/*} is encoded in a pattern $S$ as $S(1) = H(\texttt{"a"})$, $S(2) = H(\texttt{"b"})$, and all other slots free. Given the private key for $S$, one can use \enc{}'s $\mathbf{KeyDer}$ to fill in slots $3 \ldots \ell_1$. This allows one to generate the private key for \path{a/b}, \path{a/b/c}, etc.---any URI for which \path{a/b} is a prefix. To grant access to only a specific resource (a full URI, not a prefix), the \texttt{\$} is included as before.

In encoding a time range into a pattern, single
timestamps (e.g., granting access for an hour) are done as before. The hierarchical structure for time makes it possible to succinctly grant permission for an entire day, month, or year. For example, one may grant access for all of 2017 by filling in slot $\ell_2$ with $H(\texttt{"2017"})$ and leaving the final $\ell_2 - 1$ slots, which correspond to month, day, and year, free. Therefore, to grant permission over a time range, \emph{the number of keys granted is logarithmic in the length of the time range}. For example, to delegate access to a URI from October 29, 2014 at 10 PM until December 2, 2014 at 1 AM, the following keys need to be generated: \path{2014/Oct/29/23}, \path{2014/Oct/29/24}, \path{2014/Oct/30/*}, \path{2014/Oct/31/*}, \path{2014/Nov/*}, \path{2014/Dec/01/*}, and \path{2014/Dec/02/01}.
The tree can be chosen differently to support longer time ranges (e.g., additional level representing decades), change the granularity of expiry (e.g., minutes instead of hours), trade off encryption time for key size (e.g., deeper/shallower tree), or use a more regular structure (e.g., binary encoding with logarithmic split). For example, our implementation uses a depth-6 tree (instead of depth-4), to be able to delegate time ranges with fewer keys.

In summary, to produce a key set for delegation, first determine which subtrees in the time hierarchy represent the time range. For each one, produce a separate pattern, and encode the time into the last $\ell_2$ slots. Encode the URI prefix in the first $\ell_1$ slots of each pattern. Finally, generate the keys corresponding to those patterns, using keys in the key store.

\subsection{Optimizations for Low-Power Devices}\label{s:encryption_lowpower}

On low-power embedded devices, performing a single \enc{} encryption consumes a significant amount of energy. Therefore, we design \newprot{} with optimizations to \enc{}.

\subsubsection{Hybrid Encryption and Key Reuse}\label{ssec:hybrid}

\newprot{} uses \enc{} in a hybrid encryption scheme. To encrypt a message $m$ in \newprot{}, one samples a symmetric key $k$, and encrypts $k$ with \newprot{} to produce ciphertext $c_1$. The pattern used for \enc{} encryption is chosen as in \secref{ssec:uri_time} to encode the \emph{rendezvous point}. Then, one encrypts $m$ using $k$ to produce ciphertext $c_2$. The \newprot{} ciphertext is $(c_1, c_2)$.

For subsequent messages, one reuses $k$ and $c_1$; the new message is encrypted with $k$ to produce a new $c_2$. One can keep reusing $k$ and $c_1$ until the \enc{} pattern for encryption changes, which happens at the end of each hour (or other interval used for expiry). At this time, \newprot{} performs \emph{key rotation} by choosing a new $k$, encrypting it with \enc{} using the new pattern, and then proceeding as before. Therefore, \emph{most messages only incur cheap symmetric-key encryption}.

This also reduces the load on subscribers. The \newprot{} ciphertexts sent by a publisher during a single hour will all share the same $c_1$. Therefore, the subscriber can decrypt $c_1$ once for the first message to obtain $k$, and \emph{cache} the mapping from $c_1$ to $k$ to avoid expensive \enc{} decryptions for future messages sent during that hour.

Thus, expensive \enc{} operations are only performed upon key rotation, which happens \emph{rarely}---once an hour (or other granularity chosen for expiry) for each resource.

\subsubsection{Precomputation with Adjustment}\label{s:precomputation}

Even with hybrid encryption and key reuse to perform \enc{} encryption rarely, \enc{} contributes significantly to the overall power consumption on low-power devices.
Therefore, this section explores how to perform individual \enc{} encryptions more efficiently.

Most of the work to encrypt under a pattern $S$ is in computing the quantity $Q_S = g_3 \cdot \prod_{(i, a_i) \in \fixed(S)} h_i^{a_i}$, where $g_3$ and the $h_i$ are part of the \enc{} public parameters.
One may consider computing $Q_S$ once, and then reusing its value when computing future encryptions under the same pattern $S$. Unfortunately, this alone does not improve efficiency because the pattern $S$ used in one \enc{} encryption is different from the pattern $T$ used for the next encryption.

\newprot{}, however, observes that $S$ and $T$ are similar; they match in the $\ell_1$ slots corresponding to the URI, and the remaining $\ell_2$ slots will correspond to adjacent leaves in the time tree. \newprot{} takes advantage of this by efficiently \emph{adjusting} the precomputed value $Q_S$ to compute $Q_T$ as follows:
\begin{equation*}
\vspace{-3mm}
Q_T = Q_S
\cdot
\prod_{\substack{(i, b_i) \in \fixed(T) \\ i \in \free(S)}} h_i^{b_i}
\cdot
\prod_{\substack{(i, a_i) \in \fixed(S) \\ i \in \free(T)}} h_i^{-a_i}
\cdot
\prod_{\substack{(i, a_i) \in \fixed(S) \\ (i, b_i) \in \fixed(T) \\ a_i \neq b_i}} h_i^{b_i - a_i}
\vspace{-1mm}
\end{equation*}
This requires one $\mathbb{G}_1$ exponentiation per differing slot between $S$ and $T$ (i.e., the Hamming distance).
Because $S$ and $T$ usually differ in only the final slot of the time hierarchy, this will usually require one $\mathbb{G}_1$ exponentiation total, substantially faster than computing $Q_T$ from scratch. Additional exponentiations are needed at the end of each day, month, and year, but they can be eliminated by maintaining additional precomputed values corresponding to the start of the current day, current month, and current year.

The protocol remains secure because a ciphertext is distributed identically whether it was computed from a precomputed value $Q_S$ or via regular
\iffull
encryption as in \appref{app:enc}.
\fi

\subsection{Extensions}\label{ssec:extensions}
Via simple extensions, JEDI can support (1) wildcards in the \emph{middle} of a URI or time, and (2) forward secrecy. We describe these extensions in
\iffull
\appref{sec:extensions}.
\fi

\iffull
\subsection{Security Guarantee (Proof in \appref{app:proofs})}
\fi
\label{s:encryption_guarantee}

We formalize the security of JEDI's encryption below.
\begin{theorem}
\label{thm:newprot_game}
Suppose \newprot{} is instantiated with a Selective-ID CPA-secure~\cite{boneh2004efficient, abdalla2007generalized}, history-independent
\iffull
(\appref{app:proofs})
\fi
\enc{} scheme. Then, no probabilistic polynomial-time adversary $\mathcal{A}$ can win the following security game against a challenger $\mathcal{C}$ with non-negligible advantage:\\
\textbf{Initialization.} $\mathcal{A}$ selects a (URI, time) pair to attack.\\
\textbf{Setup.} $\mathcal{C}$ gives $\mathcal{A}$ the public parameters of the \newprot{} instance.\\
\textbf{Phase 1.} $\mathcal{A}$ can make three types of queries to $\mathcal{C}$:
\begin{enumerate}[topsep=0pt, noitemsep, wide, labelwidth=!, labelindent=0pt]
    \item $\mathcal{A}$ asks $\mathcal{C}$ to create \anent{}; $\mathcal{C}$ returns a name in $\{0, 1\}^*$, which $\mathcal{A}$ can use to refer to that \ent{} in future queries. A special name exists for the authority.
    \item $\mathcal{A}$ asks $\mathcal{C}$ for the key set of any \ent{}; $\mathcal{C}$ gives $\mathcal{A}$ the keys that the \ent{} has. At the time this query is made, the requested key may \textbf{not} contain a key whose URI and time are both prefixes of the challenge (URI, time) pair.
    \item $\mathcal{A}$ asks $\mathcal{C}$ to make any \ent{} delegate a key set of $\mathcal{A}$'s choice to another \ent{} (specified by names in $\{0, 1\}^*$).
\end{enumerate}
\textbf{Challenge.} When $\mathcal{A}$ chooses to end Phase 1, it sends $\mathcal{C}$ two messages, $m_0$ and $m_1$, of the same length. Then $\mathcal{C}$ chooses a random bit $b \in \{0, 1\}$, encrypts $m_b$ under the challenge (URI, time) pair, and gives $\mathcal{A}$ the ciphertext.\\
\textbf{Phase 2.} $\mathcal{A}$ can make additional queries as in Phase 1.\\
\textbf{Guess.} $\mathcal{A}$ outputs $b' \in \{0, 1\}$, and wins the game if $b = b'$.
The advantage of an adversary $\mathcal{A}$ is $\left| \Pr[\mathcal{A} \text{ wins}] - \frac{1}{2} \right|$.
\end{theorem}

We prove this theorem  in
\iffull
Appendix~\ref{app:proofs}.
\fi
Although we only achieve selective security in the standard model (like much prior work~\cite{boneh2005hierarchical, abdalla2007generalized}), one can achieve adaptive security if the hash function $H$ in \secref{ssec:key_set} is modeled as a random oracle~\cite{abdalla2007generalized}.
It is sufficient for \newprot{} to use a CPA-secure (rather than CCA-secure) encryption scheme because \newprot{} messages are signed, as detailed below in \secref{sec:integrity}.

%% file: 04_signatures.tex
\section{Integrity in \newprot{}}\label{sec:integrity}

To prevent an attacker from flooding the system with messages, spoofing fake data, or actuating devices without permission, \newprot{} must ensure that \anent{} can only send a message on a URI if it has permission. For example, an application subscribed to \path{buildingA/floor2/roomLHall/sensor0/temp} should be able to verify that the readings it is receiving are produced by \texttt{sensor0}, not an attacker. In addition to subscribers, an intermediate party (e.g., the router in a publish-subscribe system) may use this mechanism to filter out malicious traffic, without being trusted to read messages.

\subsection{Starting Solution: Signature Chains}\label{ssec:signature_chains}
A standard solution in the existing literature, used by SPKI/SDSI~\cite{clarke2001certificate}, Vanadium~\cite{taly2016distributed}, and bw2~\cite{andersen2017democratizing}, is to include a certificate chain with each message. Just as permission to subscribe to a resource is granted via a chain of delegations in \secref{sec:encryption}, permission to publish to a resource is also granted via a chain of delegations.
Whereas \secref{sec:encryption} includes \enc{} keys in each delegation, these integrity solutions delegate signed certificates. To send a message, \anent{} encrypts it (\secref{sec:encryption}), signs the ciphertext, and includes a certificate chain that proves that the signing keypair is authorized for that URI and time.

\subsection{Anonymous Signatures}\label{ssec:anon_sig}
The above solution reveals the sender's identity (via its public key) and the particular chain of delegations that gives the sender access. For some applications this is acceptable, and its auditability may even be seen as a benefit.
For other applications, the sender must be able to send a message anonymously. See \secref{s:intro_integrity} for an example.
How can we reconcile \emph{access control} (ensuring the sender has permission) and \emph{anonymity} (hiding who the sender is)?

\subsubsection{Starting Point: \enc{} Signatures}\label{s:starting_wkdibe_signatures}

Our solution is to use a signature scheme based on \enc{}. Abdalla et al.~\cite{abdalla2007generalized} observe that \enc{} can be extended to a signature scheme in the same vein as has been done for IBE~\cite{boneh2001identity} and HIBE~\cite{gentry2002hierarchical}. To sign a message $m\in\mathbb{Z}_p^*$ with a key for pattern $S$, one uses $\mathbf{KeyDer}$ to fill in a slot with $m$, and presents the decryption key as a signature.

This is our starting point for designing anonymous signatures in \newprot{}. A message can be signed by first hashing it to $\mathbb{Z}_p^*$ and signing the hash as above. Just as consumers receive decryption keys via a chain of delegations (\secref{sec:encryption}), publishers of data receive these signing keys via chains of delegations.

\subsubsection{Anonymous Signatures in \newprot{}}\label{s:opt_anon_sig}

The construction in \secref{s:starting_wkdibe_signatures} has two shortcomings. First, signatures are \emph{large}, linear in the number of fixed slots of the pattern. Second, it is unclear if they are truly \emph{anonymous}.

\parhead{Signature size}
As explained in
\iffull
\secref{sec:encryption} and \appref{app:enc},
\fi
we use a construction of \enc{} based on BBG HIBE~\cite{boneh2005hierarchical}. BBG HIBE supports a property called \emph{limited delegation} in which a secret key can be reduced in size, in exchange for limiting the depth in the hierarchy at which subkeys can be generated from it. We observe that the \enc{} construction also supports this feature. Because we need not support $\mathbf{KeyDer}$ for the decryption key acting as a signature, we use limited delegation to compress the signature to just two group elements.

\parhead{Anonymity}
The technique in \secref{s:starting_wkdibe_signatures} transforms an encryption scheme into a signature scheme, but the resulting signature scheme is not necessarily anonymous.
For the particular construction of \enc{} that we use, however, we prove that the resulting signature scheme is indeed anonymous.
Our insight is that, for this construction of \enc{}, keys are \emph{history-independent} in the following sense: $\mathbf{KeyDer}$, for a fixed $\mathsf{Params}$ and $\mathsf{Pattern_B}$, returns a private key $\mathsf{Key}_\mathsf{Pattern_B}$ with the \emph{exact same distribution} regardless of $\mathsf{Key}_\mathsf{Pattern_A}$ (see \secref{ssec:enc} for notation).
Because signatures, as described in \secref{s:starting_wkdibe_signatures}, are private keys generated with $\mathbf{KeyDer}$, they are also history-independent; a signature for a pattern has the same distribution regardless of the key used to generate it. This is precisely the anonymity property we desire.

\subsection{Optimizations for Low-Power Devices}\label{s:signature_lowpower}

As in \secref{ssec:hybrid}, we must avoid computing a \enc{} signature for every message. A simple way to do this is to sample a digital signature keypair each hour, sign the verifying key with \enc{} at the beginning of the hour, and sign messages during the hour with the corresponding signing key.

Unfortunately, this may still be too expensive for low-power embedded devices because it requires a digital signature, which requires asymmetric-key cryptography, for \emph{every} message. We can circumvent this by instead (1) choosing a \emph{symmetric} key $k$ every hour, (2) signing $k$ at the start of each hour (using \enc{} for anonymity), and (3) using $k$ in an \emph{authenticated broadcast protocol} to authenticate messages sent during the hour. An authenticated broadcast protocol, like \uTESLA{}~\cite{perrig2001spins}, generates a MAC for each message using a key whose hash is the previous key; thus, the single signed key $k$ allows the recipient to verify later messages, whose MACs are generated with hash preimages of $k$. In general, this design requires stricter time synchronization than the one based on digital signatures, as the key used to generate the MAC depends on the time at which it is sent. However, for the sense-and-send use case typical of smart buildings, sensors anyway publish messages on a fixed schedule (e.g., one sample every $x$ seconds), allowing the key to depend only on the message index. Thus, timely message delivery is the only requirement. Our scheme differs from \uTESLA{} because the first key (end of the hash chain) is signed using \enc{}.

Additionally, we use a technique similar to precomputation with adjustment (\secref{s:precomputation}) for anonymous signatures. Conceptually, $\mathbf{KeyDer}$, which is used to produce signatures, can be understood as a two-step procedure: (1) produce a key of the correct form and structure (called $\mathbf{NonDelegableKeyDer}$), and (2) re-randomize the key so that it can be safely delegated (called $\mathbf{ResampleKey}$). Re-randomization can be accelerated using the same precomputed value $Q_S$ that \newprot{} uses for encryption (\secref{s:precomputation}), which can be efficiently adjusted from one pattern to the next. The result of $\mathbf{NonDelegableKeyDer}$ can also be adjusted to obtain the corresponding result for a similar pattern more efficiently.
We fully explain our adjustment technique for signatures in
\iffull
Appendix \ref{as:precomputation}.
\fi

Finally, \enc{} signatures as originally proposed (\secref{s:starting_wkdibe_signatures}) are verified by encrypting a random message under the pattern corresponding to the signature, and then attempting to decrypt it using the key acting as a signature. We provide a more efficient signature verification algorithm for this construction of \enc{} in
\iffull
\appref{app:enc}.
\fi

\iffull
\subsection{Security Guarantee (Proof in \appref{app:proofs})}
\fi
The integrity guarantees of the method in this section can be formalized using a game very similar to the one in Theorem \ref{thm:newprot_game}, so we do not present it here for brevity.
We do, however, formalize the anonymous aspect of \enc{} signatures:
\vspace{-1ex}
\begin{theorem}
\label{thm:anonymity}
For any well-formed keys $k_1$, $k_2$ corresponding to the same (URI, time) pair in the same resource hierarchy, and any message $m \in \mathbb{Z}_p^*$, the distribution of signatures over $m$ produced using $k_1$ is information-theoretically indistinguishable from (i.e., equal to) the distribution of signatures over $m$ produced using $k_2$.
\end{theorem}
\vspace{-1ex}
This implies that even a powerful adversary who observes the private keys held by all \ents{} cannot distinguish signatures produced by different \ents{}, for a fixed message and pattern.
No computational assumptions are required.
We prove Theorem \ref{thm:anonymity} in
\iffull
\appref{app:proofs}.
\fi

%% file: 05_revocation.tex
\section{Revocation in \newprot{}}\label{sec:revocation}

This section explains how \newprot{} keys may be revoked.

\subsection{Simple Solution: Revocation via Expiry}\label{ssec:expiry_revocation}

A simple solution for revocation is to rely on expiration. In this solution, all keys are time-limited, and delegations are periodically refreshed, according to a higher layer protocol, by granting a new key with a later expiry time. In this setup, the \ent{} who granted a key can easily revoke it by not refreshing that delegation when the key expires. We expect this solution to be sufficient for many applications of \newprot{}.

\subsection{Immediate Revocation}

Some disadvantages of the solution in \secref{ssec:expiry_revocation} are that (1) \ents{} must periodically come online to refresh delegations, and (2) revocation only takes effect when the delegated key expires. We would like a solution without these disadvantages.

However, any revocation scheme that does not wait for keys to expire is subject to set of {\em inherent} limitations. The recipient of the revoked delegation still has the revoked decryption key, so it can still decrypt messages encrypted in the same way. This means that we must either (1) rely on intermediate parties to modify ciphertexts so that revoked keys cannot decrypt them, or (2) require senders to be aware of the revocation, and encrypt messages in a different way so that revoked keys cannot decrypt them. Neither solution is ideal: (1) makes assumptions about how messages are delivered, which we have avoided thus far (\secref{sec:security}), and requires trust in an intermediary to modify ciphertexts, and (2) weakens the decoupling of senders and receivers (\secref{ssec:challenges}). We adopt the second compromise: while senders will not need to know who are the receivers, they will need to know who has been revoked.

\subsection{Immediate Revocation in \newprot{}}
\label{ssec:immediate_revocation}

We extend tree-based broadcast encryption~\cite{naor2001revocation, dodis2002public} to support decentralized delegation of decryption keys, and incorporate it into \newprot{}. We use tree-based broadcast encryption because it only requires senders to know about \emph{revoked} users when encrypting messages, as opposed to \emph{all} users in the system (as is required by other broadcast encryption schemes).

\subsubsection{Tree-based Broadcast Encryption}

Existing work~\cite{naor2001revocation, dodis2002public} proposes two methods of tree-based broadcast encryption: Complete Subtree (CS) and Subset Difference (SD). We focus on the CS method here.

The CS method is based on a binary tree (\figref{fig:CS}) where each node corresponds to a separate keypair. Each user corresponds to a leaf of the tree and has the secret keys for all nodes on the root-to-leaf path. To encrypt a message that is decryptable by a subset of users, one finds a collection of subtrees that include all leaves except those corresponding to revoked users and encrypts the message multiple times using the public keys corresponding to the root of each subtree. By associating each node with an ID and encrypting with IBE, one can avoid generating a separate keypair for each node.

\subsubsection{Modifying Broadcast Encryption for Delegation}\label{s:delegable_cs}

Users in broadcast encryption do not map one-to-one to users in \newprot{}.
To avoid confusion, we refer to ``users'' in broadcast encryption as ``\leaves{}'' (abbreviated $\mathsf{\lf{}}$).

We modify the CS method to support delegation, as follows. Each key corresponds to a range of consecutive \leaves{}. When a user qualifies a key to delegate to another \ent{}, she produces a new key corresponding to a subrange of the \leaves{} of the original key. When a key is revoked, publishers are informed of the range of \leaves{} corresponding to the revoked key. Then, they encrypt new messages using the CS method, choosing subtrees that cover all leaves except those corresponding to revoked \leaves{}. If a key is revoked, that key and all keys derived from it can no longer decrypt messages, which is a property that we want.
Thus, if Alice has $k$ \leaves{}, she must store secret keys for $O(k+\log n)$ nodes, where $n$ is the total number of \leaves{} (so the depth of the tree is $\log n$).

In JEDI,  we reduce this to $O(\log n)$ secret keys by using HIBE. We give each node $v_i$ an identifier $\mathsf{id}(v_i)\in \{0,1\}^*$ that describes the path from the root of the tree to that node. In particular, if $v_j$ is an ancestor of $v_i$, then $\mathsf{id}(v_j)$ is a prefix of $\mathsf{id}(v_i)$. Note that if we use HIBE with these IDs directly, a user with the secret key for the root can generate keys for all nodes in the tree.
To fix this, we use a property called {\em limited delegation}, introduced by prior work~\cite{boneh2005hierarchical}, to generate a HIBE key that is unqualifiable (i.e., cannot be extended). For example, if Alice has \leaves{} $\mathsf{\lf{}_3}$ to $\mathsf{\lf{}_4}$ in \figref{fig:CS}, she stores an unqualifiable key for node $v_1$ and a qualifiable key for node $v_3$.
In general, each user must store $O(\log k)$ qualifiable keys and $O(\log n)$ unqualifiable keys, thus $O(\log k + \log n)$ total.

\begin{figure}[t!]
\centering
\usetikzlibrary{trees}
\scalebox{0.9}{
\begin{tikzpicture}[
    edge from parent fork down,
  every node/.style={fill=blue!30,rounded corners},
  red/.style={fill=red!60,rounded corners},
  green/.style={fill=red!60,rounded corners},
  yellow/.style={fill=green!60,rounded corners},
  edge from parent/.style={black,thick,draw},
  level distance=1cm,
    level 1/.style={sibling distance=4cm},
    level 2/.style={sibling distance=2cm},
    level 3/.style={sibling distance=1cm}]
        \node [green] {$v_1:sk_1$}
        child {
            node [yellow] {$v_2:sk_2$}
            child {
                node {$v_4:sk_4$}
                child {
                    node {$\mathsf{\lf{}_1}$}
                }
            }
            child {
                node {$v_5:sk_5$}
                child {
                    node {$\mathsf{\lf{}_2}$}
                }
            }
        }
        child {
            node [red] {$v_3:sk_3$}
            child {
                node [red] {$v_{6}:sk_{6}$}
                child {
                    node [red] {$\mathsf{\lf{}_3}$}
               }
            }
            child {
                node [red] {$v_{7}:sk_{7}$}
                child {
                    node [red] {$\mathsf{\lf{}_4}$}
               }
            }
        }
    ;
\end{tikzpicture}
}
\caption{Key management of the CS method. Red nodes indicate nodes associated with revoked leaves. The green node is the root of the subtree covering unrevoked leaves.}
\label{fig:CS}
\vspace{-3ex}
\end{figure}
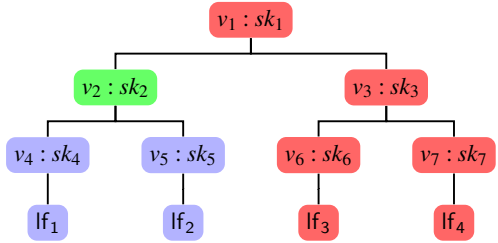

\subsubsection{Using Delegable Broadcast Encryption in \newprot{}}
\label{ssec:using_revocation}
Secret keys in our modified broadcast encryption scheme consist of HIBE keys, so incorporating it into \newprot{} is simple. As discussed in \secref{ssec:concurrent_hierarchies}, \newprot{} uses \enc{} in a way that provides multiple concurrent hierarchies, each in the vein of HIBE. Therefore, we can instantiate a third hierarchy of depth $\ell_3 = \log n$ and use it for revocation.

Let $r$ be the number of revoked keys. The CS method has $O(r \log\frac{n}{r})$-size ciphertexts, so \newprot{} ciphertexts grow to this size when revocation is used. When encrypting a message, senders use the same encryption protocol from \secref{sec:encryption} for the first $\ell_1 + \ell_2$ slots, and repeat the process, filling in the remaining $\ell_3$ slots with the ID of each node used for broadcast encryption. The size of secret keys is $O(\log k + \log n)$ after our modifications to the CS method, so \newprot{} keys grow by this factor, to a total of $O((\log k + \log n) \cdot \log T)$ \enc{} keys, where $T$ is the length of the time range for expiry.

The construction in this section works to revoke decryption keys, but cannot be used with anonymous signatures (\secref{ssec:anon_sig}). Extensions of tree-based broadcast encryption to signatures exist~\cite{libert2012group, libert2012scalable}, and we expect them to be useful to develop a construction for anonymous signatures.

How can \newprot{} inform publishers which \leaves{} are revoked? One simple option is to have a global revocation list, which \ents{} can append to. However, storing this information in a single list becomes a central point of attack, which we have avoided in our system thus far (\secref{sec:security}). To avoid this, one can store the revocation list in a global-scale blockchain, such as Bitcoin or Ethereum, which would require an adversary to be exceptionally powerful to mount a successful attack.
When revoking a set of \leaves{}, \anent{} uses those keys to sign a predetermined object (as in \secref{ssec:anon_sig}), proving it owns an ancestor of that key in the hierarchy. To keep the revocation list private, one can use \newprot{}'s encryption to ensure that only \ents{} with permission to publish to a particular resource can see which keys are revoked for that resource (since publishers too have signing keys, as described in \secref{sec:integrity}).

\iffull
\subsection{Security Guarantee (Proof in \appref{app:proofs})}
\fi

The security guarantee for immediate revocation can be stated as a modification to the game in Theorem \ref{thm:newprot_game}. In the Initialization Phase, when $\mathcal{A}$ gives $\mathcal{C}$ the challenge (URI, time), $\mathcal{A}$ additionally submits a list of revoked \leaves{}. Furthermore, $\mathcal{A}$ may compromise \ents{} in possession of private keys that can decrypt the challenge (URI, time) pair during Phases 1 and 2, as long as all \leaves{} corresponding to those keys are in the revocation list submitted in the Initialization Phase.
We provide a proof in
\iffull
\appref{app:proofs}.
\fi

\subsection{Optimizing \newprot{}'s Immediate Revocation}\label{s:optimizing_revocation}

A single \newprot{} ciphertext, with revocation enabled, consists of $O(r \log\frac{n}{r})$ \enc{} ciphertexts. To compute them efficiently, we observe that there is a large overlap in the patterns used in individual \enc{} encryptions, allowing us to use the ``precomputation with adjustment'' strategy from \secref{s:precomputation}.

Even with the above optimization, immediate revocation substantially increases the cost of \newprot{}'s cryptography. To reduce this cost, we make three observations. First, to extend \newprot{}'s hybrid encryption to work with revocation, it is sufficient to additionally rotate keys whenever the revocation list changes, in addition to the end of each hour (as in \secref{ssec:hybrid}). This means that, in the common case where the revocation list does not change in between two messages, efficient symmetric-key encryption can be used. Second, the revocation list used to encrypt a message need only contain revoked \leaves{} for the \emph{particular URI} to which the message is sent. This not only makes the broadcast encryption more efficient (smaller $r$), but also causes the effective revocation list for a stream of data to change even more rarely, allowing \newprot{} to benefit more from hybrid encryption. Third, we can do the same thing as above using the expiry time rather than the URI, allowing us to \emph{cull} the revocation list by removing keys from it once they expire.

The efficiency of hybrid encryption depends on the revocation list changing \emph{rarely}. We believe this is a reasonable assumption; most revocation will be handled by expiry, so immediate revocation is only needed if \anent{} must lose access \emph{unexpectedly}. In the smart buildings use case (\secref{sec:introduction}), for example, a key would need to be revoked if \anent{} unexpectedly transfers to another job.

The SD method for tree-based broadcast encryption can also be extended to support delegation and incorporated into \newprot{}
\iffull
(\appref{sec:sd}),
\fi
The SD method has smaller ciphertexts but larger keys.

%% file: 06_implementation.tex
\section{Implementation}\label{sec:implementation}

We implemented \newprot{} as a library in the Go programming language.
\iffull
We expect that only a few applications will require the anonymous signature protocol in \secref{ssec:anon_sig} or the tree-based revocation protocol in \secref{ssec:immediate_revocation}; most applications can use signature chains (\secref{ssec:signature_chains}) for integrity and expiry for revocation (\secref{ssec:expiry_revocation}). Therefore, our implementation makes anonymous signatures optional and implements revocation separately.
\fi
We expect \newprot{}'s key delegation to be computed on relatively powerful devices, like laptops, smartphones, or Raspberry Pis; less powerful devices (e.g., right half of \figref{fig:iot}) will primarily send and receive messages, rather than generate keys for delegation. Therefore, our focus for low-power platforms was on the ``sense-and-send'' use case~\cite{brunelli2014povomon, dutta2007procrastination, feldmeier2009personalized} typical of indoor environmental sensing, where a device periodically publishes sensor readings to a URI.
Whereas our Go library provides higher-level abstractions, we expect low-power devices to use \newprot{}'s crypto library directly.

\subsection{C/C++ Library for \newprot{}'s Cryptography}\label{s:cryptolib}

As part of \newprot{}, we implemented a cryptography library optimized in assembly for three different architectures typical of IoT platforms (\figref{fig:iot}). It implements \enc{} and \newprot{}'s optimizations and modifications
\iffull
(in \secref{s:encryption_lowpower}, \secref{s:signature_lowpower}, \appref{app:enc}).
\fi
The construction of \enc{} is based on a bilinear group in which the Bilinear Diffie-Hellman Exponent assumption holds.
\iffull
We originally planned to use Barreto-Naehrig elliptic curves~\cite{kawahara2016barreto, cheon2006security}
to implement \enc{}. Unfortunately, a recent attack on Barreto-Naehrig curves~\cite{kim2016extended}
reduced their estimated security level from 128 bits to at most 100 bits~\cite{barbulescu2017updating}.
Therefore, we
\fi
use the recent BLS12-381 elliptic curve~\cite{bowe2018bls12}.

State-of-the-art cryptography libraries implement BLS12-381, but none of them, to our knowledge, optimize for microarchitectures typical of low-power embedded platforms. To improve energy consumption, we implemented BLS12-381 in C/C++, profiled our implementation, and re-wrote performance-critical routines in assembly.
We focus on ARM Cortex-M, an IoT-focused family of 32-bit microprocessors typical of contemporary low-power embedded sensor platforms~\cite{hamiltoniot, campbell2017hail, imix}.
Cortex-M processors have been used in billions of devices, including commercial IoT offerings such as Fitbit and Nest Protect. Our assembly targets Cortex-M0+, which is among the least powerful of processors in the Cortex-M series, and of those used in IoT devices (farthest to the right in \figref{fig:iot}). By demonstrating the practicality of \newprot{} on Cortex-M0+, we establish that \newprot{} is viable across the spectrum of IoT devices (\figref{fig:iot}).

The main challenge in targeting Cortex-M0+ is that the 32-bit multiply instruction provides only the lower 32 bits of the product. Even on more powerful microarchitectures without this limitation (e.g., Intel Core i7), most CPU time ($\geq 80\%$) is spent on multiply-intensive operations (e.g., BigInt multiplication and Montgomery reduction), so the lack of such an instruction was a performance bottleneck.
As a workaround, our assembly code emulates multiply-accumulate with carry in 23 instructions. Cortex-M3 and Cortex-M4, which are more commonly used than Cortex-M0+, have instructions for 32-bit multiply-accumulate which produce the entire 64-bit result; we expect \newprot{} to be more efficient on those processors.

We also wrote assembly to optimize BLS12-381 for x86-64 and ARM64, representative of server/laptop and smartphone/Raspberry Pi, respectively (first two tiers in \figref{fig:iot}).
Thus, our Go library, which runs on these non-low-power platforms, also benefits from low-level assembly optimizations.

\subsection{Application of \newprot{} to bw2}\label{sssec:bw2}

We used our \newprot{} library to implement end-to-end encryption in bw2, a syndication and authorization system for IoT. bw2's syndication model is based on publish-subscribe, explained in \secref{sec:introduction}. Here we discuss bw2's authorization model.
Access to resources is granted via certificate chains from the authority of a resource hierarchy to \anent{}.
Individual certificates are called Declarations of Trust (DOTs).
bw2 maintains a publicly accessible registry of DOTs, implemented using blockchain smart contracts, so that \ents{} can find the DOTs they need to form DOT chains. A \emph{trusted} router enforces permissions granted by DOTs. \Ents{} must present DOT chains when publishing/subscribing to resources, and the router verifies them.
Note that a compromised router can read messages.

We use \newprot{} to enforce bw2's authorization semantics with end-to-end encryption. DOTs granting permission to subscribe now contain \enc{} keys to decrypt messages. By default, DOTs granting permission to publish to a URI remain unchanged, and are used as in \secref{ssec:signature_chains}. \enc{} keys may also be included in DOTs granting publish permission, for anonymous signatures (\secref{ssec:anon_sig}).
Using our library for \newprot{}, we implemented a wrapper around the bw2 client library. It transparently encrypts and decrypts messages using \enc{}, and includes \enc{} parameters and keys in DOTs and \ents{}, as needed for \newprot{}. bw2 signs each message with a digital signature (first alternative in \secref{s:signature_lowpower}).

The bw2-specific wrapper is less than 900 lines of Go code. Our implementation required no changes to bw2's client library, router, blockchain, or core---it is a separate module. Importantly, it provides the same API as the standard bw2 client library. Thus, it can be used as a drop-in replacement for the standard bw2 client library, to easily add end-to-end encryption to existing bw2 applications with minimal changes.

%% file: 07_evaluation.tex
\section{Evaluation}\label{sec:evaluation}

We evaluate \newprot{} via microbenchmarks, determine its power consumption on a low-power sensor, measure the overhead of applying it to bw2, and compare it to other systems.

\begin{table}[t]
    \centering
    \caption{Latency of \newprot{}'s implementation of BLS12-381}
    \begin{tabular}{|l|c|c|c|}\hline
        \textbf{Operation} & \textbf{Laptop} & \hspace{-1ex}\textbf{Rasp. Pi}\hspace{-1ex} & \textbf{Sensor}\\\hline\hline
        $\mathbb{G}_1$ Mul. (Chosen Scalar) & 109 \us{} & 1.33 ms & 509 ms\\\hline
        $\mathbb{G}_2$ Mul. (Chosen Scalar) & 343 \us{} & 3.86 ms & 1.44 s\\\hline
        $\mathbb{G}_T$ Mul. (Rand. Scalar) & 504 \us{} & 5.47 ms & 1.90 s\\\hline
        $\mathbb{G}_T$ Mul. (Chosen Scalar) & 507 \us{} & 5.48 ms & 2.81 s\\\hline
        Pairing & 1.29 ms & 14.0 ms & 4.99 s\\\hline
    \end{tabular}
    \label{tab:bls12_381}
    \vspace{-4ex}
\end{table}

\subsection{Microbenchmarks}\label{ssec:microbenchmarks}

Benchmarks labeled ``Laptop'' were produced on a Lenovo T470p laptop with an Intel Core i7-7820HQ CPU @ 2.90 GHz. Benchmarks labeled ``Raspberry Pi'' were produced on a Raspberry Pi 3 Model B+ with an ARM Cortex-A53 @ 1.4 GHz. Benchmarks labeled ``Sensor'' were produced on a commercially available ultra low-power environmental sensor platform called ``Hamilton'' with an ARM Cortex-M0+ @ 48 MHz. We describe Hamilton in more detail in \secref{ssec:iot_eval}.

\subsubsection{Performance of BLS12-381 in \newprot{}}

Table \ref{tab:bls12_381} compares the performance of \newprot{}'s BLS12-381 implementation on the three platforms, with our assembly optimizations. As expected from \figref{fig:iot}, the Raspberry Pi performance is an order of magnitude slower than Laptop performance, and performance on the Hamilton sensor is an additional two-to-three orders of magnitude slower.

\subsubsection{Performance of \enc{} in \newprot{}}

\figref{fig:crypto_performance} depicts the performance of \newprot{}'s cryptography primitives. \figref{fig:crypto_performance} does not include the sensor platform; \secref{ssec:iot_eval} thoroughly treats performance of \newprot{} on low-power sensors.

In \figref{sfig:wkdibe}, we used a pattern of length 20 for all operations, which would correspond to, e.g., a URI of length 14 and an Expiry hierarchy of depth 6.
To measure decryption and signing time, we measure the time to decrypt the ciphertext or sign the message, plus the time to generate a decryption key for that pattern or ID. For example, if one receives a message on \path{a/b/c/d/e/f}, but has the key for \path{a/*}, he must generate the key for \path{a/b/c/d/e/f} to decrypt it.

\figref{sfig:wkdibe} demonstrates that the \newprot{} encrypts and signs messages and generates qualified keys for delegation at practical speeds. On a laptop, all \enc{} operations take less than 10 ms with up to 20 attributes. On a Raspberry Pi, they are 10x slower (as expected), but still run at interactive speeds.

\begin{figure}[t]
    \centering
    \begin{subfigure}[p]{0.455\linewidth}
        \begin{tabular}{|l|c|c|}\hline
             & \hspace{-0.5ex}Laptop\hspace{-0.5ex} & \hspace{-0.5ex}Rasp. Pi\hspace{-0.5ex} \\\hline\hline
            \hspace{-0.5ex}Enc.\hspace{-0.5ex} & \hspace{-0.5ex}3.08 ms\hspace{-0.5ex} & \hspace{-0.5ex}37.3 ms\hspace{-0.5ex}\\\hline
            \hspace{-0.5ex}Dec.\hspace{-0.5ex} & \hspace{-0.5ex}3.61 ms\hspace{-0.5ex} & \hspace{-0.5ex}43.9 ms\hspace{-0.5ex}\\\hline
            \hspace{-0.5ex}KeyD.\hspace{-0.5ex} & \hspace{-0.5ex}4.77 ms\hspace{-0.5ex} & \hspace{-0.5ex}58.5 ms\hspace{-0.5ex}\\\hline
            \hspace{-0.5ex}Sign\hspace{-0.5ex} & \hspace{-0.5ex}4.80 ms\hspace{-0.5ex} & \hspace{-0.5ex}61.2 ms\hspace{-0.5ex}\\\hline
            \hspace{-0.5ex}Verify\hspace{-0.5ex} & \hspace{-0.5ex}4.78 ms\hspace{-0.5ex} & \hspace{-0.5ex}56.3 ms\hspace{-0.5ex}\\\hline
        \end{tabular}
        \caption{Latency of $\mathbf{Encrypt}$, $\mathbf{Decrypt}$, $\mathbf{KeyDer}$, $\mathbf{Sign}$, and $\mathbf{Verify}$ with 20 attributes}
        \label{sfig:wkdibe}
    \end{subfigure}
    \begin{subfigure}[p]{0.505\linewidth}
        \includegraphics[width=\linewidth]{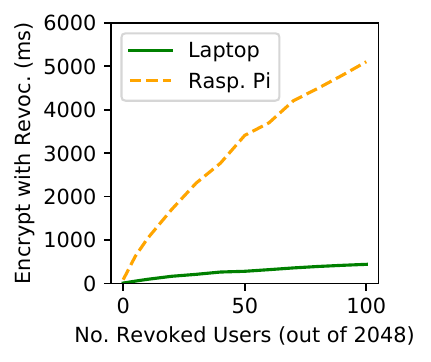}
        \caption{Encryption with Revocation}
        \label{sfig:revocation_encryption}
    \end{subfigure}
    \caption{Performance of \newprot{}'s cryptography}
    \label{fig:crypto_performance}
    \vspace{-3ex}
\end{figure}

\subsubsection{Performance of Immediate Revocation in \newprot{}}

\figref{sfig:revocation_encryption} shows the cost of \newprot{}'s immediate revocation protocol (\secref{sec:revocation}). A private key containing $k$ \leaves{} consists of $O(\log k + \log n)$ \enc{} secret keys where $n$ is the total number of \leaves{}. Therefore, the performance of immediate revocation depends primarily on the number of \leaves{}.

To encrypt a message, one \enc{} encryption is performed for each subtree needed to cover all unrevoked \leaves{}. In general, encryption is $O(r\log\frac{n}{r})$, where $r$ is the number of revoked \leaves{}. Each key contains a set of \emph{consecutive} \leaves{}, so encryption is also $O(R\log\frac{n}{R})$, where $R$ is the number of revoked \newprot{} keys. Decryption time remains almost the same, since only one \enc{} decryption is needed.

To benchmark revocation, we use a complete binary tree of depth 16 ($n = 65536$). The time to generate a new key for delegation is essentially independent of the number of \leaves{} conveyed in that key, because $\log k \ll \log n$. We empirically confirmed this; the time to generate a key for delegation was constant at 2.4 ms on a laptop and 31 ms on a Raspberry Pi as the number of \leaves{} in the key was varied from 5 to 1,000.

To benchmark encryption with revocation, we assume that there exist 2,048 users in the system each with 32 \leaves{}. We measure encryption time with a pattern with 20 fixed slots (for URI and time) as we vary the number of revoked users. \figref{sfig:revocation_encryption} shows that encryption becomes expensive when the revocation list is large ($500$ milliseconds on laptop and $\approx 5$ seconds on Raspberry Pi). However, such an encryption only needs to be performed by a publisher when the URI, time, or revocation list changes; subsequent messages can reuse the underlying symmetric key (\secref{s:optimizing_revocation}). Furthermore, the revocation list includes only revoked keys that match the (URI, time) pair being used, so it is not expected to grow very large.

\begin{figure}[t]
    \centering
    \begin{subfigure}[p]{0.50\linewidth}
        \includegraphics[width=\linewidth]{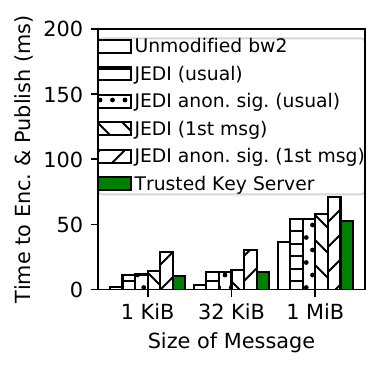}
        \caption{Encrypt/publish message}
        \label{sfig:publish}
    \end{subfigure}
    \begin{subfigure}[p]{0.48\linewidth}
        \includegraphics[width=\linewidth]{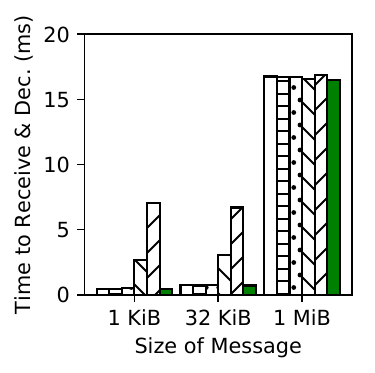}
        \caption{Receive/decrypt message}
        \label{sfig:subscribe}
    \end{subfigure}
    \caption{Critical-path operations in bw2, with/without \newprot{}}
    \label{fig:critical_bw2}
    \vspace{-3ex}
\end{figure}

\subsection{Performance of \newprot{} in bw2}\label{ssec:real}

In bw2, the two critical-path operations are publishing a message to a URI, and receiving a message as part of a subscription. We measure the overhead of \newprot{} for these operations because they are core to bw2's functionality and would be used by any messaging application built on bw2. Our methodology is to perform each operation repeatedly in a loop, to measure the sustained performance (operations/second), and report the average time per operation (inverse).
To minimize the effect of the network, the router was on the same link as the client, and the link capacity was 1 Gbit/s.
In our experiments, we used a URI of length 6 and an Expiry tree of depth 6. We also include measurements from a strawman system with pre-shared AES keys---this represents the critical-path overhead of an approach based on the Trusted Key Server discussed in \secref{sec:security}. Our results are in \figref{fig:critical_bw2}.

We implement the optimizations in \secref{ssec:hybrid}, so only symmetric key encryption/decryption must be performed in the common case (labeled ``usual'' in the diagram).
However, the symmetric keys will \emph{not} be cached for the first message sent every hour, when the \enc{} pattern changes. A \enc{} operation must be performed in this case (labeled ``1st message'' in the diagram).
For large messages, the cost of symmetric key encryption dominates.
\newprot{} has a particularly small overhead for 1 MiB messages in \figref{sfig:subscribe}, perhaps because 1 MiB messages take several milliseconds to transmit over the network, allowing the client to decrypt a message while the router is sending the next message.

We also consider creating DOTs and initiating subscriptions, which are not in the critical path of bw2.
These results are in \figref{fig:occasional_bw2} (note the log scale in \figref{sfig:dot}). Creating DOTs is slower with \newprot{}, because \enc{} keys are generated and included in the DOT.
Initiating a subscription in bw2 requires forming a DOT chain; in \newprot{}, one must also derive a private key from the DOT chain.
\figref{sfig:dot} shows the time to form a short one-hop DOT chain, and in the case of \newprot{}, includes the time to derive the private key.
For \newprot{}'s encryption (\secref{sec:encryption}), these additional costs are incurred only by DOTs that grant permission to subscribe. With anonymous signatures, DOTs granting permission to publish incur this overhead as well, as \enc{} keys must be included.
\figref{sfig:init} puts this in context by measuring the end-to-end latency from initiating a subscription to receiving the first message (measured using bw2's ``query'' functionality).

\begin{figure}[t]
    \centering
    \vspace{-1ex}
    \begin{subfigure}[p]{0.48\linewidth}
        \includegraphics[width=\linewidth]{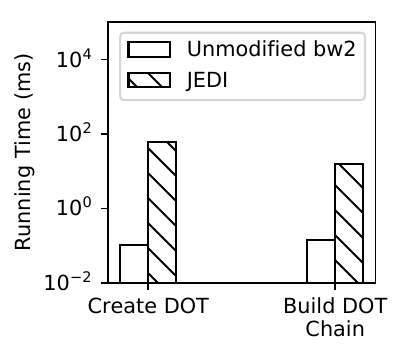}
        \caption{Create DOT, Build Chain}
        \label{sfig:dot}
    \end{subfigure}
    \begin{subfigure}[p]{0.48\linewidth}
        \includegraphics[width=\linewidth]{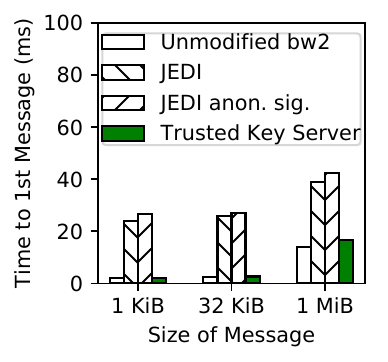}
        \caption{Time to Query/Subscribe}
        \label{sfig:init}
    \end{subfigure}
    \caption{Occasional bw2 operations, with and without \newprot{}}
    \label{fig:occasional_bw2}
    \vspace{-3ex}
\end{figure}

For a DOT to be usable, it must be inserted into bw2's registry. This requires a blockchain transaction (not included in \figref{fig:occasional_bw2}). An important consideration in this regard is \emph{size}.
In the unmodified bw2 system, a DOT that grants permission on \path{a/b/c/d/e/f} is 198 bytes. With \newprot{}, each DOT also contains multiple \enc{} keys, according to the time range. In the ``worst case,'' where the start time of a DOT is Jan 01 at 01:00, and the end time is Dec 31 at 22:59, a total of 45 keys are needed. Each key is approximately
\iffull
1 KiB (Table \ref{tab:size}),
\fi
so the size of this DOT is approximately 45 KiB.

Because bw2's registry of DOTs is implemented using blockchain smart contracts, the bandwidth for inserting DOTs is limited. Using \newprot{} would increase the size of DOTs as above, resulting in an approximately 100-400x decrease in aggregate bandwidth for creating DOTs. However, this can be mitigated by changing bw2 to not store DOTs directly in the blockchain. DOTs can be stored in untrusted storage, with only their hashes stored in the blockchain-based registry.
Such a solution could be based on Swarm~\cite{tron2016smash} or Filecoin~\cite{filecoin}.

\subsection{Feasibility on Ultra Low-Power Devices}\label{ssec:iot_eval}

We use a commercially available sensor platform called ``Hamilton''~\cite{hamiltoniot, andersen2017hamilton} built around the Atmel SAMR21 system-on-chip (SoC). The SAMR21 costs approximately \$2.25 per unit~\cite{samr21e18a} and integrates a low-power microcontroller and radio.
The sensor platform we used in this study costs \$18 to manufacture~\cite{kim2018system}. For battery lifetime calculations, we assume that the platform is powered using a CR123A Lithium battery that provides 1400 mAh at 3.0 V (252 J of energy). Such a battery costs \$1.
The SAMR21 is heavily constrained: it has only a 48 MHz CPU frequency based on the ARM Cortex-M0+ microarchitecture, and a total of only 32 KiB of data memory (RAM). Our goal is to validate that \newprot{} is practical for an ultra low-power sensor platform like Hamilton, in the context of a ``sense-and-send'' application in a smart building. Since most of the platform's cost (\$18) comes from the on-board transducers and assembly, rather than the SAMR21 SoC, \emph{using an even more resource-constrained SoC would not significantly decrease the platform's cost.} An analogous argument applies to energy consumption, as the transducers account for more than half of Hamilton's idle current~\cite{kim2018system}.

Hamilton/SAMR21 is on the lower end of platforms typically used for sense-and-send applications in buildings. Some older studies~\cite{feldmeier2009personalized, li2014energy} use even more constrained hardware like the TelosB; this is because those studies were constrained by hardware available at the time. Modern 32-bit SoCs, like the SAMR21, offer substantially better performance at a similar price/power point to those older platforms~\cite{kim2018system}.

\subsubsection{CPU Usage}

Table \ref{tab:cpu_power} shows the time for encryption and anonymous signing in \newprot{} on Hamilton. The results use the optimizations discussed in \secref{s:encryption_lowpower} and \secref{s:signature_lowpower}, and include the time to ``adjust'' precomputed state. They indicate that symmetric keys can be encrypted and anonymously signed in less than 10 seconds. This is feasible given that encryption and anonymous signing occur rarely, once an hour, and need not be produced at interactive speeds in the normal ``sense-and-send'' use case.

\begin{table}[t]
    \centering
    \caption{CPU and power costs on the Hamilton platform}
    \begin{tabular}{|l|c|c|} \hline
        \textbf{Operation} & \textbf{Time} & \hspace{-1ex}\textbf{Average Current}\hspace{-1ex} \\ \hline\hline
        Sleep (Idle) & N/A & 0.0063 mA \\ \hline
        \enc{} Encrypt & 6.50 s & 10.2 mA \\ \hline
        \enc{} Encrypt and Sign & 9.89 s & 10.2 mA \\ \hline
    \end{tabular}
    \label{tab:cpu_power}
    \vspace{-2ex}
\end{table}

\subsubsection{Power Consumption}

To calculate the impact on battery lifetime, we consider a ``sense-and-send'' application, in which the Hamilton device obtains readings from its sensors at regular intervals, and immediately sends the readings encrypted over the wireless network. We measured the average current consumed for varying sample intervals, when each message is encrypted with AES-CCM, without using JEDI (``AES Only'' in Table \ref{tab:lifetime}). We estimate \newprot{}'s average current based on the current, duration, and frequency (once per hour, for these estimates) of \newprot{} operations, and add it to the average current of the ``AES Only'' setup. Our estimates assume that the \uTESLA{}-based technique in \secref{s:signature_lowpower} is used to avoid attaching a digital signature to each message. We divide the battery's energy capacity by the result to compute lifetime. As shown in Table \ref{tab:lifetime}, \newprot{} decreases battery life by about 40-60\%. Battery life is several years even with \newprot{}, acceptable for IoT sensor platforms.

\newprot{}'s overhead depends primarily on the granularity of expiry times (one hour, for these estimates), \emph{not} the sample interval. To improve power consumption, one could use a time tree with larger leaves, allowing \ents{} to perform \enc{} encryptions and anonymous signatures less often. This would, of course, make expiry times coarser.

\begin{table}[t]
    \centering
    \caption{Average current and expected battery life (for 1400 mAh battery) for sense-and-send, with varying sample interval}
    \begin{tabular}{|l|c|c|c|} \hline
         & \textbf{AES Only} & \textbf{\newprot{} (enc)} & \hspace{-1ex}\textbf{\newprot{} (enc \& sign)}\hspace{-1ex} \\ \hline
        10 s & 32 \uA{} / 5.1 y & 50 \uA{} / 3.2 y & 60 \uA{} / 2.6 y \\ \hline
        20 s & 20 \uA{} / 8.1 y & 38 \uA{} / 4.2 y & 48 \uA{} / 3.3 y \\ \hline
        30 s & 15 \uA{} / 10 y & 34 \uA{} / 4.7 y & 44 \uA{} / 3.6 y \\ \hline
    \end{tabular}
    \label{tab:lifetime}
    \vspace{-3ex}
\end{table}

\subsubsection{Memory Budget}

Performing \enc{} operations requires only 6.5 KiB of data memory, which fits comfortably within the 32 KiB of data memory (RAM) available on the SAMR21. The code space required for our implementation of \enc{} and BLS12-381 is about 74 KiB, which fits comfortably in the 256 KiB of code memory (ROM) provided by the SAMR21.

A related question is whether storing a hash chain in memory (as required for authenticated broadcast, \secref{s:signature_lowpower}) is practical. If we use a granularity of 1 minute for authenticated broadcast, the length of the hash chain is 60. At the start of an hour, one computes the entire chain, storing 10 hashes equally spaced along the chain, each separated by 5 hashes. As one progresses along the hash chain, one re-computes each set of 5 hashes one additional time. This requires storage for only 15 hashes ($<4$ KiB memory) and computation of only 105 hashes \emph{per hour}, which is practical. One could possibly optimize performance further using \emph{hierarchical hash chains}~\cite{hu2005efficient}.

\subsubsection{Impact of \newprot{}'s Optimizations}
\newprot{}'s cryptographic optimizations (\secref{s:precomputation}, \secref{s:opt_anon_sig}, \secref{s:signature_lowpower}), which use \enc{} in a non-black-box manner, provide a 2-3x performance improvement. Our assembly optimizations (\secref{sec:implementation}) provide an additional 4-5x improvement. Without both of these techniques, \newprot{} would not be practical on low-power sensors.
Hybrid encryption and key reuse (\secref{ssec:hybrid}), which let \newprot{} use \enc{} \emph{rarely}, are also crucial.

\subsection{Comparison to Other Systems}\label{s:comparison}

\definecolor{darkgreen}{RGB}{0,128,0}
\newcommand{\mplus}{\textbf{\textcolor{darkgreen}{+}}}
\newcommand{\mmin}{\textbf{\textcolor{red}{--}}}

\begin{table*}[t]
    \centering
    \caption{Comparison of \newprot{} with other crypto-based IoT/cloud systems}
    {\setlength{\tabcolsep}{1ex}
    \begin{tabular}{|p{1.1in}||p{0.7in}|p{2.1in}|p{2.55in}|}\hline
       {\bf  Crypto Scheme / System} & {\bf Avoids Central Trust?} & {\bf Expressivity} & {\bf Performance} \\\hline\hline
        Trusted Key Server (\secref{sec:security}) &
        \tablebullets{
            \item[\mmin] No
        } &
        \tablebullets{
            \item[\mplus] Supports arbitrary policies (beyond hierarchies)
            \item[\mmin] No delegation
        } &
        \tablebullets{
            \item[\mplus] $\approx 10$ \us{} to encrypt 1 KiB message (same as \newprot{} in common case, faster for first message after key rotation)
            \item[\mmin] Trusted party generates one key \emph{per resource}
        }\\\hline
        PRE (Lattice-Based), as used in PICADOR \cite{borcea2017picador} &
        \tablebullets{
            \item[\mmin] No
        } &
        \tablebullets{
            \item[\mplus] Supports arbitrary policies (beyond hierarchies)
            \item[\mmin] No delegation
        } &
        \tablebullets{
            \item[\mplus] $\approx 5$ ms encrypt, $\approx 3$ ms decrypt (similar to \newprot{}: 3-4 ms)
            \item[\mmin] Trusted party must generate one key per sender-receiver pair
        }\\\hline
        PRE (Pairing-Based), as used in Pilatus \cite{shafagh2017secure} &
        \tablebullets{
            \item[\mplus] Yes
        } &
        \tablebullets{
            \item[\mmin] Delegation is single-hop
            \item[\mmin] Delegation is coarse (all-or-nothing)
            \item[\mplus] Can compute aggregates on encrypted data
        } &
        \tablebullets{
            \item[\mplus] 0.6 ms encrypt, 1.3 ms re-encrypt, 0.5 ms decrypt (faster than \newprot{}: 3-4 ms)
            \item[\mplus] Practical on constrained IoT device with crypto accelerator
        }\\\hline
        CP-ABE~\cite{bethencourt2007ciphertext} &
        \tablebullets{
            \item[\mplus] Yes
        } &
        \tablebullets{
            \item[\mplus] Good fit for RBAC policies
            \item[\mmin] Cannot support \newprot{}'s hierarchy abstraction with delegation
        } &
        \tablebullets{
            \item[\mplus] Only symmetric crypto in common case
            \item[\mmin] 14 ms encrypt for first time after key rotation (4-5x slower than \newprot{}: 3 ms)
        }\\\hline
        KP-ABE, as used in Sieve~\cite{wang2016sieve} &
        \tablebullets{
            \item[\mplus] Yes
        } &
        \tablebullets{
            \item[\mplus] Succinct delegation based on attributes
            \item[\mmin] Delegation is single-hop
        } &
        \tablebullets{
            \item[\mplus] Only symmetric crypto in common case
            \item[\mmin] 25 ms encrypt for first time after key rotation (8-9x slower than \newprot{}: 3 ms)
        }\\\hline
        Delegable Large Univ. KP-ABE~\cite{goyal2006attribute} (used in Alternative \newprot{} Design) &
        \tablebullets{
            \item[\mplus] Yes
        } &
        \tablebullets{
            \item[\mplus] Generalizes beyond hierarchies and supports multi-hop delegation (subsumes \newprot{})
        } &
        \tablebullets{
            \item[\mplus] Only symmetric crypto in common case
            \item[\mmin] 60 ms encrypt for first time after key rotation (20x slower than \newprot{}: 3 ms)
            \item[\mmin] Impractical for low-power sense-and-send
        }\\\hline\hline
       \textbf{This paper:} WKD-IBE ~\cite{abdalla2007generalized} with Optimizations, as used in \newprot{}  &
        \tablebullets{
            \item[\mplus] Yes
        } &
        \tablebullets{
            \item[\mplus] Delegation is multi-hop
            \item[\mplus] Succinct delegation of \emph{subtrees} of resources (or more complex sets, \secref{ssec:extensions})
            \item[\mplus] Non-interactive expiry
        } &
        \tablebullets{
            \item[\mplus] After key rotation (e.g., once per hour), 3 ms encrypt, 4 ms decrypt (\figref{sfig:wkdibe})
            \item[\mplus] Only symmetric crypto in common case
            \item[\mplus] Practical for ultra low-power ``sense-and-send'' \emph{without crypto accelerator}
        }\\\hline
    \end{tabular}}
    \vspace{-2ex}
    \label{tab:comparison}
\end{table*}

Table \ref{tab:comparison} compares \newprot{} to other systems and cryptographic approaches, particularly those geared toward IoT, in regard to security, expressivity and performance.
We treat these existing systems as they would be used in a messaging system for smart buildings (\secref{sec:introduction}). Table \ref{tab:comparison} contains quantitative comparisons to the cryptography used by these systems; for those schemes based on bilinear groups, we re-implemented them using our \newprot{} crypto library (\secref{s:cryptolib}) for a fair comparison.

\parhead{Security} The owner of a resource is considered \emph{trusted} for that resource, in the sense that an adversary who compromises \anent{} can read all of that \ent{}'s resources. In Table \ref{tab:comparison}, we focus on whether a single component is trusted for \emph{all} resources in the system. Note that, although Trusted Key Server (\secref{sec:security}) and PICADOR~\cite{borcea2017picador} encrypt data in flight, granting or revoking access to \anent{} requires participation of an \emph{online trusted party} to generate new keys.

\parhead{Expressivity} PRE-based approaches, which associate public keys with users and support delegation via proxy re-encryption, are fundamentally coarse-grained---a re-encryption key allows \emph{all} of a user's data to be re-encrypted. PICADOR~\cite{borcea2017picador} allows more fine-grained semantics, but does not enforce them cryptographically. ABE-based approaches typically do not support delegation beyond a single hop, whereas \newprot{} achieves multi-hop delegation. In ABE-based schemes, however, attributes/policies attached to keys can describe more complex sets of resources than \newprot{}. That said, a hierarchical resource representation is sufficient for \newprot{}'s intended use case, namely smart cities; existing syndication systems for smart cities, which do not encrypt data and are unconstrained by the expressiveness of crypto schemes, choose a hierarchical rather than attribute-based representation (\secref{sec:introduction}).

\parhead{Performance} The Trusted Key Server (\secref{sec:security}) is the most na\"ive approach, requiring an online trusted party to enforce all policy. Even so, \newprot{}'s performance in the common case is the same as the Trusted Key Server (\figref{fig:critical_bw2}), because of \newprot{}'s hybrid encryption---\newprot{} invokes \enc{} \emph{rarely}. Even when \newprot{} invokes \enc{}, its performance is not significantly worse than PRE-based approaches. An alternative design for \newprot{} uses the GPSW KP-ABE construction instead of \enc{}, but it is significantly more expensive. Based Table \ref{tab:lifetime}, the power cost of a \enc{} operation \emph{even when only invoked once per hour} contributes significantly to the overall energy consumption on the low-power IoT device; using KP-ABE instead of \enc{} would increase this power consumption by an order of magnitude, reducing battery life significantly.

\noindent
\textbf{In summary,} existing systems fall into one of three categories. (1) The Trusted Key Server allows access to resources to be managed by arbitrary policies, but relies on a \emph{central trusted party} who must be online whenever a user is granted access or is revoked. (2) PRE-based approaches, which permit sharing via re-encryption, cannot cryptographically enforce fine-grained policies or support multi-hop delegation. (3) ABE-based approaches, if carefully designed, \emph{can} achieve the same expressivity as \newprot{}, but are substantially less performant and are not suitable for low-power embedded devices.

%% file: 08_related.tex
\section{Related Work}\label{sec:related}
We organize related work into the following categories.

\parhead{Traditional Public-Key Encryption}
SiRiUS~\cite{goh2003sirius} and Plutus~\cite{kallahalla2003plutus} are encrypted filesystems based on traditional public-key cryptography, but they do not support delegable and qualifiable keys like \newprot{}.
Akl et al.~\cite{akl1983cryptographic} and further work~\cite{crampton2015cryptographic, crampton2006key}  propose using key assignment schemes for access control in a hierarchy.
A line of work \cite{tzeng2002time, huang2004new, atallah2007incorporating, atallah2009dynamic} builds on this idea to support both hierarchical structure and temporal access. Key assignment approaches, however, require the full hierarchy to be known at setup time, which is not flexible in the IoT setting. \newprot{} does not require this, allowing different subtrees of the hierarchy to be managed separately (\secref{s:overview}, ``Delegation'').

\parhead{Identity-Based Encryption}
Tariq et al.~\cite{tariq2014securing} use Identity-Based Encryption (IBE)~\cite{boneh2001identity} to achieve end-to-end encryption in publish-subscribe systems, without the router's participation in the protocol. However, their approach does not support hierarchical resources. Further, encryption and private keys are  on a credential-basis, so each message is encrypted multiple times according to the credentials of the recipients.

Wu et al.~\cite{wu2016privacy} use a prefix encryption scheme based on IBE for mutual authentication in IoT. Their prefix encryption scheme is different from \newprot{}, in that users with keys for identity \path{a/b/c} can decrypt messages encrypted with prefix identity \path{a}, \path{a/b} and \path{a/b/c}, but not identities like \path{a/b/c/d}.

\parhead{Hierarchical Identity-Based Encryption}
Since the original proposal of Hierarchical Identity-Based Encryption (HIBE)~\cite{gentry2002hierarchical}, there have been multiple
HIBE constructions~\cite{gentry2002hierarchical, boneh2004efficient, boneh2005hierarchical, gentry2009hierarchical}
and variants of HIBE~\cite{yao2004id, abdalla2007generalized}.
Although seemingly a good match for resource hierarchies, HIBE cannot be used as a black box to efficiently instantiate \newprot{}. We considered alternative designs of \newprot{} based on existing variants of HIBE, but as we elaborate in
\iffull
\appref{app:hibe},
\fi
each resulting design is either less expressive or significantly more expensive than \newprot{}.

\parhead{Attribute-Based Encryption}
A line of work~\cite{yu2010achieving,wang2016sieve} uses Attribute-Based Encryption (ABE)~\cite{goyal2006attribute, bethencourt2007ciphertext} to delegate permission.
\iffull
For example, Yu et al.~\cite{yu2010achieving} and Sieve~\cite{wang2016sieve} use Key-Policy ABE (KP-ABE)~\cite{goyal2006attribute} to control which \ents{} have access to encrypted data in the cloud.
Some of these approaches also provide a means to revoke users, leveraging proxy re-encryption to safely perform re-encryption in the cloud.
\fi
Our work additionally supports hierarchically-organized resources and decentralized delegation of keys, which \cite{yu2010achieving} and \cite{wang2016sieve} do not address.
As discussed in
\iffull
\secref{s:comparison} and \appref{app:kpabe},
\fi
\enc{} is substantially more efficient than
KP-ABE and provides enough functionality for \newprot{}.
\iffull
\enc{} could be a lightweight alternative to KP-ABE for some applications.
\fi

Other approaches prefer Ciphertext-Policy ABE (CP-ABE)~\cite{bethencourt2007ciphertext}. Existing work~\cite{wang2010hierarchical, wang2011hierarchical} combines HIBE with CP-ABE to produce Hierarchical ABE (HABE), a solution for sharing data on untrusted cloud servers. The ``hierarchical'' nature of HABE, however, corresponds to the hierarchical organization of domain managers in an enterprise, not a hierarchical organization of \emph{resources} as in our work.

\parhead{Proxy Re-Encryption}
NuCypher KMS~\cite{egorov2017nucypher} allows a user to store data in the cloud encrypted under her public key, and share it with another user using Proxy Re-Encryption (PRE)~\cite{blaze1998divertible}. While NuCypher assumes limited collusion among cloud servers and recipients (e.g., $m$ of $n$ secret sharing) to achieve properties such as expiry, \newprot{} enforces expiry via cryptography, and therefore remains secure against \emph{any} amount of collusion. Furthermore, NuCypher's solution for resource hierarchies requires a keypair for each node in the hierarchy, meaning that the creation of resources is centralized. Finally, keys in NuCypher are not qualifiable.
\iffull
Given a key for \path{a/*}, one cannot generate a key for \path{a/b/*} to give to another \ent{}.
\fi

PICADOR~\cite{borcea2017picador}, a publish-subscribe system with end-to-end encryption, uses a lattice-based PRE scheme.
However, PICADOR requires a central Policy Authority to specify access control, by creating a re-encryption key for every permitted pair of publisher and subscriber. In contrast, \newprot{}'s access control is decentralized.

\parhead{Revocation Schemes}
Broadcast encryption (BE)~\cite{naor2001revocation, dodis2002public, boneh2005collusion, boneh2006fully, lewko2010revocation,
boneh2014low, boneh2017multiparty} is a mechanism to achieve revocation, by encrypting messages such that they are only decryptable by a specific set of users. However, these existing schemes do not support key qualification and delegation, and therefore, cannot be used in \newprot{} directly.
Another line of work builds revocation directly into the underlying cryptography primitive, achieving Revocable IBE~\cite{boldyreva2008identity, libert2009adaptive, seo2013revocable, watanabe2017new},
Revocable HIBE~\cite{seo2013efficient, seo2015revocable, liu2016compact} and Revocable KP-ABE~\cite{attrapadung2009conjunctive}. These papers use a notion of revocation in which URIs are revoked. In contrast, \newprot{} supports revocation at the level of keys. If multiple \ents{} have access to a URI, and one of their keys is revoked, then the other \ent{} can still use its key to access the resource.
Some systems~\cite{egorov2017nucypher, belguith2018secure} rely on the participation of servers or routers to achieve revocation.

\parhead{Secure Reliable Multicast Protocol}
Secure Reliable Multicast~\cite{malkhi2000secure, malkhi1997high} also uses a many-to-many communication model, and ensures correct data transfer in the presence of malicious routers. \newprot{}, as a protocol to \emph{encrypt} messages, is complementary to those systems.

\parhead{Authorization Services}
\newprot{} is complementary to authorization services for IoT, such as bw2~\cite{andersen2017democratizing}, Vanadium~\cite{taly2016distributed}, WAVE~\cite{andersen2019wave}, and AoT~\cite{neto2016aot}, which focus on expressing authorization policies and enabling \ents{} to prove they are authorized, rather than on encrypting data.
Droplet~\cite{shafagh2018droplet} provides encryption for IoT, but does not support delegation beyond one hop and does not provide hierarchical resources.

An authorization service that provides secure in-band permission exchange, like WAVE~\cite{andersen2019wave}, can be used for key distribution in \newprot{}. \newprot{} can craft keys with various permissions, while WAVE can distribute them without a centralized party by including them in its attestations.

%% file: 09_conclusion.tex
\section{Conclusion}\label{sec:conclusion}

In this paper, we presented \newprot{}, a protocol for end-to-end encryption for IoT.
\newprot{} provides \emph{many-to-many} encrypted communication on complex resource hierarchies, supports decentralized  key delegation, and decouples senders from receivers. It provides  expiry for access to resources, reconciles anonymity and authorization via anonymous signatures, and allows revocation via tree-based broadcast encryption. Its encryption and integrity solutions are capable of running on embedded devices with strict energy and resource constraints, making it suitable for the Internet of Things.

%% file: a_wkdibe.tex
\section{\newprot{}'s Optimizations to \enc{}}\label{app:enc}

The purpose of \appref{app:enc} is twofold. First, it largely reproduces the construction of \enc{} provided in \S{}3.2 of \cite{abdalla2007generalized}, which we used in \newprot{}, for readers who would like additional context. Second, it fully explains the ``Precomputation with Adjustment'' optimizations (\secref{s:precomputation}) that \newprot{} uses for fast encryption and signatures on low-power platforms.

\subsection{Construction with our Optimizations}\label{ssec:enc_construction}

We present the \enc{} construction in \S{}3.2 of \cite{abdalla2007generalized}. For completeness, we include the extension to signatures described in \secref{ssec:anon_sig}, as well as our observation in \secref{ssec:enc} that the encryption algorithm in \cite{abdalla2007generalized} can be extended to work with arbitrary patterns.

The construction is based on bilinear groups. $\mathbb{G}$ and $\mathbb{G}_T$ are cyclic groups of prime order $p$, and they are related by a bilinear map $e: \mathbb{G} \times \mathbb{G} \rightarrow \mathbb{G}_T$. The security parameter $\secp$ is related to the number of bits of $p$. We implemented \enc{} using BLS12-381, an \emph{asymmetric} bilinear group whose bilinear map is of the form $e: \mathbb{G}_1 \times \mathbb{G}_2 \rightarrow \mathbb{G}_T$. We denote $\mathbb{G}_1$ the smaller and faster of the two source groups (elliptic curve over $\mathbb{F}_p$). The construction of \enc{} was originally defined for symmetric bilinear groups. Our description of the construction below shows how we mapped the construction onto an asymmetric bilinear group.

\noindent
$\mathbf{Setup}(1^\ell)$: Select $g \stackrel{\$}{\leftarrow} \mathbb{G}_2$ and $g_2, g_3, h_1, \ldots, h_\ell, h_s \stackrel{\$}{\leftarrow} \mathbb{G}_1$. Then select $\alpha \stackrel{\$}{\leftarrow} \mathbb{Z}_p$ and let $g_1 = g^\alpha$. Output:

\noindent
$\mathsf{Params} = (g, g_1, g_2, g_3, h_1, \ldots, h_\ell, h_s)$ and $\mathsf{MasterKey} = g_2^\alpha$.

\noindent
$\mathbf{KeyDer}(K, S)$: If $K$ is the master key, take $K = g_2^\alpha$. Select $r \stackrel{\$}{\leftarrow} \mathbb{Z}_p$. The private key for the pattern $S$ is the following triple:
\begin{equation*}
\left(
g_2^\alpha \cdot \left(g_3 \cdot \prod_{(i, a_i) \in \fixed(S)} h_i^{a_i}\right)^r,\quad
g^r,\quad
 \left\{(j, h_j^r)  \right\}_{j \in \free(S)} \right).
\end{equation*}
If $K$ is not the master key, then parse $K$ as $(k_0, k_1, B)$, where $B=\{(i, b_i)\}$. Select $t \sample \mathbb{Z}_p$. The private key for $S$ is:
\begin{equation*}
\begin{split}
\Biggl(
k_0 \cdot \left(g_3 \cdot \prod_{(i, a_i) \in \fixed(S)} h_i^{a_i}\right)^t \cdot \prod_{\substack{(i, a_i) \in \fixed(S) \\ (i, b_i) \in B}} b_i^{a_i},\quad
g^t \cdot k_1,\\
\left\{(j, h_j^t \cdot b_j)\right\}_{j \in \free(S)}
\Biggr).
\end{split}
\end{equation*}

\noindent
Observe that the resulting key is identically distributed, regardless of whether or not the input key $K$ is the master key.

\noindent
$\mathbf{Encrypt}(S, m)$: Here, $m \in \mathbb{G}_T$. Select $s \stackrel{\$}{\leftarrow} \mathbb{Z}_p$ and output
$$\left(
e(g_1, g_2)^s \cdot m,\quad
g^s,\quad
\left(g_3 \cdot \prod_{(i, a_i) \in \fixed(S)} h_i^{a_i}\right)^s
\right).$$

\noindent
$\mathbf{Decrypt}(K, C)$: Parse the key $K$ as $(k_0, k_1, B)$, and the ciphertext $C$ as $(X, Y, Z)$. Output
$$X \cdot e(k_1, Z) \cdot e(Y, k_0)^{-1}.$$

To support encryption over arbitrary patterns (\secref{ssec:enc}), we only compute the product over fixed slots, just as is done in the BBG HIBE construction~\cite{boneh2005hierarchical}. In contrast, the original \enc{} construction requires all slots in the pattern to be fixed, and iterates over all slots. The proof technique from \cite{boneh2005hierarchical}, namely padding the selected ID with zeros, can be used here to modify the proof of \enc{}~\cite{abdalla2007generalized} to account for our optimization that allows free slots to be used in encryption.

\subsection{Construction of \enc{} Signatures}\label{as:anon_sig}

We originally explained anonymous signatures in \secref{s:starting_wkdibe_signatures} by (1) dedicating one of the slots in a pattern for signing messages, and (2) defining signature generation as a call to $\mathbf{KeyDer}$ that fills that dedicated slot. We then proposed improvements to the signature scheme, to make it constant size (\secref{s:opt_anon_sig}) and make the verification procedure more efficient (\secref{s:signature_lowpower}). We formally describe our optimizations to the signature algorithm below. Note that we added an extra term to the public parameter $h_s$ that represents the slot dedicated to signing messages, and an analogous element ($s$, $b_s$) to the third component of each secret key. It was not present in the original \enc{} construction, and is not used for encryption in \appref{ssec:enc_construction}.

\noindent
$\mathbf{Sign}(K, m)$: Parse the key $K$ as $(k_0, k_1, B)$, where $(s, b_s) \in B$. Let $S$ be the pattern corresponding to $K$. Select $t \sample \mathbb{Z}_p$ and output
$$\left(
k_0 \cdot \left(g_3 \cdot h_s^m \cdot \prod_{(i, a_i) \in \fixed(S)} h_i^{a_i}\right)^t \cdot b_s^m,\quad
g^t \cdot k_1
\right)$$

\noindent
$\mathbf{Verify}(S, \sigma, m)$: Parse the signature $\sigma$ as $(s_0, s_1)$. Check:
$$e(s_0, g) \stackrel{?}{=} e(g_1, g_2) \cdot e\left(g_3 \cdot h_s^m \cdot \prod_{(i, a_i) \in \fixed(S)} h_i^{a_i},\quad s_1\right)$$

In contrast to optimized procedures above, the na\"{i}ve signature algorithm has $\mathbf{Sign}(K, m) = \mathbf{KeyDer}(K, T)$, and $\mathbf{Verify}(S, \sigma, m) = (\mathbf{Decrypt}(\sigma, \mathbf{Encrypt}(T, m^*)) \stackrel{?}{=} m^*)$ for $m^* \sample \mathbb{Z}_p^*$, where $T$ is the same as $S$ except that $T(s) = m$, the message being signed. The modification we make is that (1) the signature contains only the first two components of $\mathbf{KeyDer}(K, T)$ (since the third component is not used for decryption), and (2) the verification procedure checks that $\sigma$ is a private key corresponding to $T$ more efficiently than encrypting and decrypting a random message.

Finally, note that the $\textbf{Sign}$ function can be generalized to allow a key with pattern $P$ to produce a signature for pattern $S$ if $P$ matches $S$. This can be done trivially by first applying $\textbf{KeyDer}$ to obtain a key for $S$, and calling the $\textbf{Sign}$ on the existing key. Our implementation supports this \textbf{GeneralizedSign} functionality more efficiently, as follows:

\noindent
$\mathbf{GeneralizedSign}(K, S, m)$: Parse the key $K$ as $(k_0, k_1, B)$, where $(s, b_s) \in B$. Select $t \sample \mathbb{Z}_p$ and output
$$\left(
k_0 \cdot \left(g_3 \cdot h_s^m \cdot \hspace{-1ex} \prod_{(i, a_i) \in \fixed(S)} \hspace{-1ex} h_i^{a_i}\right)^t  \cdot b_s^m \cdot \hspace{-1ex} \prod_{\substack{(i, a_i) \in \fixed(S) \\ (i, b_i) \in B}} \hspace{-1ex} b_i^{a_i},\quad
g^t \cdot k_1
\right)$$

\subsection{Precomputation with Adjustment}\label{as:precomputation}

We formally explain the precomputation optimization introduced in \secref{s:precomputation} and \secref{s:signature_lowpower}.

\subsubsection{Precomputation with Adjustment for Encryption}

We define the new \enc{} operations as follows (as before, $\mathsf{Params}$ is an implicit parameter):

\noindent
$\mathbf{Precompute}(S)$: Output
$$g_3 \cdot \prod_{(i, a_i) \in \fixed(S)} h_i^{a_i}$$

\noindent
$\mathbf{AdjustPrecomputed}(Q_S, S, T)$: $Q_S$ is the existing precomputed value, $S$ is the pattern it corresponds to, and $T$ is the pattern whose precomputed value to compute. Output
$$
Q_S
\cdot
\prod_{\substack{(i, b_i) \in \fixed(T) \\ i \in \free(S)}} h_i^{b_i}
\cdot
\prod_{\substack{(i, a_i) \in \fixed(S) \\ i \in \free(T)}} h_i^{-a_i}
\cdot
\prod_{\substack{(i, a_i) \in \fixed(S) \\ (i, b_i) \in \fixed(T) \\ a_i \neq b_i}} h_i^{b_i - a_i}
$$

\noindent
$\mathbf{EncryptPrepared}(Q_S, m)$: Here, $m \in \mathbb{G}_T$. Select $s \stackrel{\$}{\leftarrow} \mathbb{Z}_p$ and output
$$\left(
e(g_1, g_2)^s \cdot m,\quad
g^s,\quad
Q_S^s
\right).$$

The above routines are used as described in \secref{s:precomputation}. Observe that the output of $\mathbf{Encrypt}(S, m)$ and the output of $\mathbf{EncryptPrepared}(\mathbf{Precompute}(S), m)$ are distributed identically---security of this optimization relies on this fact.

\subsubsection{Precomputation with Adjustment for Signatures}

As explained in \secref{s:signature_lowpower}, we think of $\mathbf{KeyDer}$ in two parts, $\mathbf{NonDelegableKeyDer}$ and $\mathbf{ResampleKey}$:

\noindent
$\mathbf{NonDelegableKeyDer}(K, S)$: Parse $K$ as $(k_0, k_1, B)$, where $B=\{(i, b_i)\}$. Output:
$$\left(
k_0 \cdot \prod_{\substack{(i, a_i) \in \fixed(S) \\ (i, b_i) \in B}} b_i^{a_i},\quad
k_1,\quad
\left\{(j, b_j)\right\}_{j \in \free(S)}
\right).$$

\noindent
$\mathbf{ResampleKey}(K, S)$: Parse $K$ as $(k_0, k_1, B)$. Sample $t \sample \mathbb{Z}_p$ and output
$$\left(
k_0 \cdot \left(g_3 \cdot \prod_{(i, a_i) \in \fixed(S)} h_i^{a_i}\right)^t,
g^t \cdot k_1,
\left\{(j, h_j^t \cdot b_j)\right\}_{j \in \free(S)}
\right).$$

$\mathbf{NonDelegableKeyDer}$ and $\mathbf{ResampleKey}$ are the ``two parts'' of $\mathbf{KeyDer}$ in the sense that $\{\mathbf{KeyDer}(K, S)\} = \{\mathbf{ResampleKey}(\mathbf{NonDelegableKeyDer}(K, S), S)\}$ where the distributions are over the sampled randomness.

We can take advantage of these functions to accelerate signing of messages. Note the similarity between $\mathbf{Sign}$ and $\mathbf{ResampleKey}$. The setup we consider is that \anent{} has a key for some pattern $R$ representing a URI prefix and time prefix. It will repeatedly sign messages with a pattern $S$ representing at a fully-qualified URI and specific time, where $R$ matches $S$. The next signature will be on pattern $T$ which shares the same URI as $S$ but corresponds to the next leaf in the time tree. The na\"ive algorithm is to call $\mathbf{QualifyKey}$ to obtain a key for $S$ and then call $\mathbf{Sign}$. The key idea behind the optimization is to instead call $\mathbf{NonDelegableKeyDer}$ to obtain a pseudo-key for $S$ (which is not safe to delegate), and then create a signature for that. Observe that the resulting signature is distributed in exactly the same way whether the na\"ive or optimized method is used.
(\secref{s:signature_lowpower}, which explains $\mathbf{Sign}$ as simply being a call to $\mathbf{KeyGen}$, explains this technique as splitting $\mathbf{KeyDer}$ into two parts, $\mathbf{NonDelegableKeyDer}$ and $\mathbf{ResampleKey}$.)

Now that we have described the signature process in terms of two calls, one to $\mathbf{NonDelegableKeyDer}$ and another to $\mathbf{Sign}$, we describe how to apply Precomputation with Adjustment to each of these operations. $\mathbf{Sign}$ can be accelerated using the same precomputed value we used to accelerate encryption. We have already shown how to ``adjust'' this precomputed value from $S$ to $T$.

\noindent
$\mathbf{SignPrepared}(K, Q_S, m)$: Parse the key $K$ as $(k_0, k_1, B)$. Let $S$ be the pattern corresponding to $K$; $Q_S$ must be the precomputed value corresponding to $S$. Select $t \sample \mathbb{Z}_p^*$ and output:
$$\left(
k_0 \cdot \left(h_s^m \cdot Q_S\right)^t,\quad
g^t \cdot k_1
\right)$$

Finally, we explain how the result of $\mathbf{NonDelegableKeyDer}$ can be adjusted from pattern $S$ to pattern $T$. The procedure also requires the parent key (whose pattern we denote $R$), on which $\mathbf{NonDelegableKeyDer}$ was called to obtain the key corresponding to pattern $S$.

\noindent
$\mathbf{AdjustNonDelegable}(P, C, S, T)$: Parse the parent key $P$ as $P$ as $(p_0, p_1, B)$ where $B = \{(i, b_i)\}$. Parse the child key $C$ as $C = (k_0, k_1, Z)$. $S$ is the pattern corresponding to $C$, and $T$ is the pattern that the resulting key will correspond to. Output:
\begin{equation*}
\begin{split}
\Biggl(
k_0 \cdot \prod_{\substack{(i, t_i) \in \fixed(T) \\ i \in \free(S)}} b_i^{t_i} \cdot \prod_{\substack{(i, s_i) \in \fixed(S) \\ i \in \free(T)}} b_i^{-s_i} \cdot \prod_{\substack{(i, s_i) \in \fixed(S) \\ (i, t_i) \in \fixed(T)}} b_i^{t_i-s_i},\\
k_1,\quad
\left\{(j, b_j)\right\}_{j \in \free(T)}
\Biggr).
\end{split}
\end{equation*}

To sign a message each hour, \newprot{} maintains the result of $\mathbf{Precompute}$, $Q_S$ (as it does for encryption), and also the result of $\mathbf{NonDelegableKeyGen}$, $C$, derived from its key. Then it adjusts both values, using $\mathbf{AdjustPrecomputed}$ and $\mathbf{AdjustNonDelegable}$, when the pattern used to sign changes. To sign a message $m$, it computes $\mathbf{SignPrepared}(C, Q_S, m)$.

As an additional optimization, we only compute the first two elements of the output of $\mathbf{NonDelegableKeyDer}$ and $\mathbf{AdjustNonDelegableKeyDer}$ when using it to produce signatures.

%% file: b_sd.tex
\section{Revocation in \newprot{} using the SD Method}\label{sec:sd}
The SD algorithm for tree-based broadcast encryption is introduced in prior work \cite{naor2001revocation, dodis2002public}. In this section, we explain how to extend SD algorithm to support delegation. We follow the same approach to apply this extension in \newprot{} as we did for the CS method in \secref{sec:revocation}.

\subsection{Overview of SD Method}
We first provide a brief overview of the SD method, as described in \cite{dodis2002public}. Both the CS and SD method use a technique called \emph{Subset Cover}. Let $N$ be the (finite) set of all users. A family of subsets $\mathcal{S}$ is defined over $N$. Each each element $S_i$ of $\mathcal{S}$ is a subset of $N$, and corresponds to a keypair ($\mathsf{PK}_{S_i}$, $\mathsf{SK}_{S_i}$). Each user is given the secret key $\mathsf{SK}_{S_i}$ for every set $S_i$ containing that user.
Now, suppose that someone wants to encrypt a message such that it is visible only to a subset of users, denoted $B$. To achieve this, she must find a \emph{subset cover} for $B$: sets $B_1, \ldots, B_p \in \mathcal{S}$ such that $B = \bigcup_{j = 1}^p B_i$. Then, she encrypts her message under the public key $\mathsf{PK}_{B_j}$ for each $B_j$ in the subset cover. Only users in $B$ have the secret key for one of the $B_j$ to decrypt the message.

The CS method (as described in \secref{ssec:immediate_revocation}) uses a subset family $\mathcal{S}$, where each subset is a complete subtree of the binary tree over the users.
The SD method uses a different subset family $\mathcal{S}$ than the CS method, in which each \leaf{} belongs to more subsets. Each subset $S_{ij} \in \mathcal{S}$ is defined in terms of two nodes, $v_i$ and $v_j$, where $v_j$ is a descendant of $v_i$. $S_{ij}$ contains all \leaves{} in the subtree rooted at $v_i$ but not in the subtree rooted at $v_j$. Each \leaf{} now belongs to $\sum_{k=1}^{\log(n)}(2^k-k)$ different subsets, and therefore has $O(n)$ secret keys.

Each subset $S_{ij}$ is associated with an ID in HIBE as follows. The first component of the ID is $\mathsf{ID}(v_i)$. The remaining components of the ID are $\mathsf{ID}(v_j)$, where each bit occupies one component, which makes it possible to generate the private key for $S_{ik}$ from the private key for $S_{ij}$, as long as $v_k$ is a descendant of $v_j$.

If a \leaf{} belongs to both $S_{ij}$ and $S_{ik}$, it only needs to be given $S_{ij}$. With this optimization, each \leaf{} only has $O(\log^2 n)$ HIBE secret keys.

For encryption, the sender must find a subset cover over all unrevoked \leaves{}, as in the CS method. The algorithm for finding the optimal subset cover is introduced in \cite{naor2001revocation}. \cite{naor2001revocation} also proves that only $O(r)$ subsets (and therefore $O(r)$ encryptions) are needed.

\subsection{Extension of SD Method for Delegation}
We follow the same idea in \secref{ssec:immediate_revocation} that each key corresponds to a range of consecutive \leaves{}, and the delegation and revocation follow the same rule. In contrast to CS method, we use \enc{} instead of HIBE to reduce the storage requirement.

\subsubsection{Analysis of Private Key Storage}

\Anent{} with a set of consecutive \leaves{} denoted $\mathsf{\LF{}}$ must be able to generate private keys for all subsets $S_{jk}\in \mathcal{S}$ where $S_{jk}\cap \mathsf{\LF{}}\neq \varnothing$.

We define $S_j$ as the set of \leaves{} in the subtree rooted at $v_j$, and $\mathsf{copath}(v)$ as in a Merkle Tree~\cite{merkle1989certified}. Then define $\mathsf{CoPath(\LF{})}$ as follows:
$$\mathsf{CoPath(\LF{})}=\mathsf{copath}(\mathsf{leaf(\lf{}_{min})})\cup\mathsf{copath}(\mathsf{leaf(\lf{}_{max})})$$
\noindent
where $\mathsf{\lf{}_{min}}$ stands for the first \leaf{} in $\mathsf{\LF{}}$ and $\mathsf{\lf{}_{max}}$ stands for the last \leaf{} in $\mathsf{\LF{}}$. $\mathsf{LeftChild}(v_i)$ and $\mathsf{RightChild}(v_i)$ stands for the left child node and the right child node of $v_i$ respectively.

Then \anent{} with $\mathsf{\LF{}}$ only needs to store the private keys for subsets $S_{jk}$ satisfying the following properties:
\begin{itemize}[leftmargin=*]
    \item $S_{jk}\cap \mathsf{\LF{}}\neq \varnothing$ and
    \begin{enumerate}[leftmargin=*]
        \item If $S_{\mathsf{LeftChild}(v_j)}\cap \mathsf{\LF{}}\neq \varnothing$ and $S_{\mathsf{RightChild}(v_j)}\cap \mathsf{\LF{}}\neq \varnothing$, then $v_k$ must satisfy $v_k=\mathsf{LeftChild}(v_j)$ or $v_k=\mathsf{RightChild}(v_j)$.
        \item If $S_{\mathsf{LeftChild}(v_j)}\cap \mathsf{\LF{}}= \varnothing$ or $S_{\mathsf{RightChild}(v_j)}\cap \mathsf{\LF{}}= \varnothing$, then $v_k$ must satisfy $v_k\in \mathsf{CoPath(\LF{})}$.
    \end{enumerate}
\end{itemize}

There will be $O(k)$ subsets $S_{jk}$ satisfying the first condition, and $O(\log^2 n)$ subsets $S_{jk}$ satisfying the second condition. Thus the total number of keys will be $O(k+\log^2 n)$. Note that we leverage HIBE as in the original SD Method to achieve $O(\log^2 n)$ keys for the second kind of subset.

However, we cannot use HIBE to reduce storage for the first kind of subsets because for two subsets $S_{ab}$ and $S_{cd}$ in the first type, either $v_a\neq v_c$ or $v_b$ is the sibling of $v_d$ (no hierarchical relations). Recall that HIBE can only help reduce storage when $v_a=v_c$ and $v_b$ is the ancestor or descendant of $v_d$. However, we observe that $v_a$ might be the ancestor or descendant of $v_c$, which means we can use \enc{} to support another concurrent hierarchy. We discuss this in \appref{s:sd_opt}.

\subsubsection{Optimization}\label{s:sd_opt}
We will represent the $O(k)$ subsets satisfying the first condition using only $O(\log k)$ keys. We will do this by introducing a hierarchy for the first index (the $i$ in $S_{ij}$). However, we are already using a hierarchy for the second index (the $j$ in $S_{ij}$), and HIBE cannot support two hierarchies simultaneously. Therefore, we will use \enc{} to combine the two hierarchies, as we did in \secref{ssec:concurrent_hierarchies}.

We use a \enc{} system with $2\log(n) + 1$ slots. We use the first $\log(n) + 1$ slots to support a hierarchy for the first index, and the remaining $\log n$ slots to support a hierarchy for the second index. Each subset $S_{ij}$ is associated with a pattern in \enc{} as follows. The first hierarchy encodes $\mathsf{ID}(v_i)$ (each bit in a separate slot), followed by \texttt{\$}, a termination symbol. The second hierarchy encodes $\mathsf{ID}(v_j)$, where each bit is stored in a separate slot.

With this construction, we preserve all of the functionality from before. The second hierarchy representing the second index can be extended, just as HIBE in original SD method. As before, the first index cannot be modified, because the terminator symbol \texttt{\$} prevents the first hierarchy from being extended. Alternatively, limited delegation can be used for this purpose (as in \secref{ssec:immediate_revocation}), to avoid using one component of the \enc{} pattern for the terminator symbol \texttt{\$}. Note that there are two logical hierarchies in this construction, and only one of them needs to be made unqualifiable via limited delegation.

Now, we can more efficiently represent the secret keys for the $O(k)$ subsets satisfying the first condition above. Observe that the $k$ consecutive \leaves{} can be grouped into $O(\log k)$ consecutive subtrees. Let $\mathsf{TR(\LF{})}$ be the set containing the root nodes of these subtrees. For each node $v_i \in \mathsf{TR(\LF{})}$, we must provide the private key corresponding to $S_{i\mathsf{LeftChild}(v_i)}$ and $S_{i\mathsf{RightChild}(v_i)}$, without including the terminator symbol \texttt{\$} in the attribute set. From these $O(\log k)$ \enc{} keys, \anent{} can obtain all $O(k)$ keys corresponding to subsets $S_{ij}$ satisfying the first condition above. This reduces the total number of secret keys for $k$ consecutive \leaves{} from $O(k + \log^2 n)$ to $O(\log k + \log^2 n)$.

\subsubsection{Using Delegable SD Method in \newprot{}} As in \secref{ssec:using_revocation}, we use the first $\ell_1+\ell_2$ slots for URI and expiry, and use the remaining slots to instantiate the above revocation protocol.
Different from the delegable CS method (\secref{s:delegable_cs}), the number of slots used for revocation is $\ell_3=2\log(n)+1$, and they are used by the delegable SD method as two independent hierarchies. So the total number of \enc{} keys is $O((\log k + \log^2 n) \cdot \log T)$, where $T$ is the length of the time range for expiry.

%% file: c_hibe.tex
\section{Alternative Designs of \newprot{} using HIBE}\label{app:hibe}

We explore alternative designs we could have used for \newprot{}, focusing on existing cryptographic primitives we considered using instead of \enc{}.

\subsection{Hierarchical Identity-Based Encryption}

Given that \newprot{} represents URIs and time as hierarchies, Hierarchical Identity-Based Encryption (HIBE)~\cite{boneh2005hierarchical} may seem like a natural building block to use. We can encode a URI in an ID for HIBE, just as we did for \enc{}. For example, the URI prefix \path{a/b/*} can be encoded into an ID as \texttt{("a", "b")}. This preserves the crucial property that the private key for a URI prefix can be used to generate the private key for any URI with that prefix. The same thing works for expiry: for example, the timestamp June 08, 2017 at 6 AM could be encoded into an ID as \texttt{("2017", "June", "08", "06")}.

However, HIBE cannot \emph{simultaneously} support a URI hierarchy and an expiry hierarchy. A simple approach would be to concatenate the IDs. For example, the key for the URI prefix \path{a/b/*} and time prefix \path{2018/Jun/*} would have the ID \texttt{("2018", "Jun", "a", "b")}. However, this idea is flawed: only the URI can be extended, not the expiry time. The same problem applies to the URI, if we put the URI before the time in the ID. Another possible approach is to interleave the resource hierarchy and time hierarchy, using metadata to distinguish the elements. In this setup, each (resource, time) pair corresponds to multiple IDs in the HIBE system---all possible interleavings of the URI the Expiry IDs. However, for a URI of length $m$ and a time of length $n$, there are exponentially many IDs, $\frac{(m + n)!}{m! \cdot n!}$, and each message sent with that URI and time must be encrypted under all of those IDs. Therefore, this approach is infeasible.

Another strawman is to use two HIBE systems, one for URIs and one for expiry. Each message is encrypted twice, using the URI ID in the first system and again using the time ID in the second system. During delegation, each \ent{} is provided with a key from the first system for the URI, and a set of keys from the second system for the time range. The problem is that this approach is not collusion-resistant: \anent{} who is given two delegations, one for the correct URI that has expired, and one for the wrong URI that has not expired, can decrypt messages by combining keys from different delegations.

\subsection{Variants of HIBE other than \enc{}}

Existing work~\cite{yao2004id} has proposed extending HIBE to MHIBE, which supports ID-based encryption for multiple concurrent hierarchies. We could use MHIBE in \newprot{} to combine the URI hierarchy with the expiry hierarchy. However, the proposed MHIBE schemes are significantly less performant than \enc{}: for two hierarchies, they have quadratically-sized private keys and ciphertexts. Size and performance degrade exponentially in the number of hierarchies. Furthermore, a formal treatment of MHIBE is not provided.

Another extension is forward secure HIBE~\cite{yao2004id, boneh2005hierarchical}, or fs-HIBE for short. The BBG construction of fs-HIBE~\cite{boneh2005hierarchical} has linear size and performance. We considered using its mechanism for forward security \emph{in reverse}, to achieve expiry with HIBE. However, the fs-HIBE construction has linear-size ciphertexts and linear-time decryption, whereas \enc{} has constant-size ciphertexts and constant-time decryption.
In the context of a real system, this is important: $\mathbf{Encrypt}$ and $\mathbf{Decrypt}$ are used in the critical path, so encryption time, decryption time, and ciphertext size must be as small as possible. In contrast, $\mathbf{Delegate}$ is only used occasionally, so the size of private keys is less important.

Most importantly, \enc{} is a more powerful primitive than either MHIBE or fs-HIBE. In particular, \enc{} supports the \texttt{+} wildcard for URIs and timestamps (\appref{ssec:plus_wildcard}), which MHIBE and fs-HIBE do not.

%% file: d_kpabe.tex
\section{Building Block Comparison: KP-ABE}\label{app:kpabe}

In Key-Policy Attribute-Based Encryption (KP-ABE), a message is encrypted with a set of attributes. An attribute set is like a string of bits; each attribute is either present in the set (1) or not present (0). Private keys are generated with an \emph{access tree}, which can be thought of as a circuit. A private key can decrypt a message if its access tree, evaluated on the bits representing the attribute set of the message, evaluates to 1.

We are interested in KP-ABE with two properties:
\begin{enumerate}[leftmargin=*, noitemsep]
    \item \textbf{Delegable.} Given the private key for an access tree, one can generate a private key for a more restrictive access key, and delegate it to another \ent{}.
    \item \textbf{Large Universe.} The space of attributes $\mathcal{A}$ is exponentially large in the security parameter $\secp$. This is similar to Identity-Based Encryption (IBE)~\cite{boneh2001identity}, as any string of bytes can be hashed to an attribute.
\end{enumerate}

The GPSW construction~\cite{goyal2006attribute} of KP-ABE, based on bilinear groups, satisfies these properties. In fact, KP-ABE with these two properties subsumes \enc{}. A pattern $T$ in \enc{} can be converted to an attribute set in Delegable Large Universe KP-ABE by hashing each non-$\bot$ component of $T$, concatentated with its index, to an attribute in KP-ABE. Private keys in \enc{} can be expressed as an access tree consisting of a single many-input AND gate.\footnote{Ciphertext-Policy ABE (CP-ABE) is not suitable for this construction. This is because attributes cannot be added to secret keys during delegation, as per the security guarantees of CP-ABE.}

The subsections below compare BBG HIBE, \enc{}, and GPSW KP-ABE. Although we include HIBE for the sake of comparison, note that HIBE is not expressive enough to realize the \newprot{} protocol (as explained in \appref{app:hibe}).

\subsection{Performance Comparison}

We compare the performance of KP-ABE, \enc{}, and HIBE in terms of the number of exponentiations and pairings, the most expensive operations in the elliptic curves. This is shown in Table \ref{tab:complexity}. $\ell$ is the total number of attributes that can be used for a single message (the implicit argument to $\mathbf{Setup}$). For Encrypt, Decrypt and $\text{KeyDer}^1$, $r$ is the number of attributes of the key or ciphertext. For $\text{KeyDer}^2$, $n$ is the number of attributes of the starting key, and $r$ is the number of attributes of the ending key. This shows that \enc{}'s performance is theoretically better than KP-ABE's performance. Furthermore, \enc{} is just efficient as HIBE, even though \enc{} is more expressive than HIBE. As discussed in \appref{app:hibe}, HIBE is not expressive enough to efficiently instantiate \newprot{}.

\begin{table}[t]
    \centering
    \caption{Performance comparison of HIBE, \enc{}, and KP-ABE in terms of pairings and exponentiations. We omit operations that can be precomputed once for all IDs (attribute sets) in the HIBE/\enc{}/KP-ABE system. $\mathbf{KeyDer^1}$ indicates deriving the new key from the master key, and $\mathbf{KeyDer^2}$ indicates the other case.}
    \begin{tabular}{|l|l|c|c|} \hline
    Operation & Scheme & Pairings & Exponentiations\\ \hline\hline
    Encrypt & HIBE & $0$ & $3 + r$ \\ \hline
    Encrypt & \enc{} & $0$ & $3 + r$ \\ \hline
    Encrypt & KP-ABE & $0$ & $2 + r \cdot (\ell + 3)$ \\ \hline\hline
    Decrypt & HIBE & $2$ & $\leq r$ \\ \hline
    Decrypt & \enc{} & $2$ & $\leq r$ \\ \hline
    Decrypt & KP-ABE & $r + 1$ & $2r$ \\ \hline\hline
    $\text{KeyDer}^1$  & HIBE & $0$ & $\ell + 2$ \\ \hline
    $\text{KeyDer}^1$  & \enc{} & $0$ & $\ell + 2$\\ \hline
    $\text{KeyDer}^1$  & KP-ABE & $0$ & $r \cdot (\ell + 5)$ \\ \hline\hline
    $\text{KeyDer}^2$ & HIBE & 0 & $(r - n) + \ell + 2$ \\ \hline
    $\text{KeyDer}^2$ & \enc{} & $0$ & $(r - n) + \ell + 2$ \\ \hline
    $\text{KeyDer}^2$ & KP-ABE & $0$ & $2n + r \cdot (\ell + 5)$ \\ \hline
    \end{tabular}
    \label{tab:complexity}

    \vspace{-1ex}
\end{table}

\subsection{Size Comparison}

We list the size of ciphertexts and private keys in Table \ref{tab:size}: $r$ is the number of attributes in the ciphertext or private key, and $\ell$ is the maximum number of slots or attributes used to encrypt a message. Note that ciphertexts in \enc{} are constant size, whereas ciphertexts in KP-ABE are linear.

\begin{table}[t]
    \centering
    \caption{Size comparison of HIBE, \enc{}, and KP-ABE in terms of number of group elements. For elliptic curves that we used, elements of $\mathbb{G}_1$ are 48 B each, elements of $\mathbb{G}_2$ are 96 B each, and elements of $\mathbb{G}_T$ are 576 B each.}
    \begin{tabular}{|l|l|c|c|c|} \hline
    Object & Scheme & $\mathbb{G}_1$ & $\mathbb{G}_2$ & $\mathbb{G}_T$\\ \hline\hline
    Ciphertext & HIBE & $1$ & $1$ & $1$ \\ \hline
    Ciphertext & \enc{} & $1$ & $1$ & $1$ \\ \hline
    Ciphertext & KP-ABE & $r$ & $1$ & $1$ \\ \hline\hline
    Private Key & HIBE & $\ell - r + 1$ & $1$ & $0$ \\ \hline
    Private Key & \enc{} & $\ell - r + 1$ & $1$ & $0$ \\ \hline
    Private Key & KP-ABE & $r$ & $r$ & $0$ \\ \hline
    \end{tabular}
    \label{tab:size}
    \vspace{-1ex}
\end{table}

%% file: e_proof.tex
\newcommand{\mpk}{\ensuremath{\mathsf{mpk}}}
\newcommand{\msk}{\ensuremath{\mathsf{msk}}}
\newcommand{\calA}{\ensuremath{\mathcal{A}}}
\newcommand{\calB}{\ensuremath{\mathcal{B}}}
\newcommand{\calC}{\ensuremath{\mathcal{C}}}
\newcommand{\calL}{\ensuremath{\mathcal{L}}}

\section{Formal Definitions and Proofs}\label{app:proofs}
In this section, we present our formal definitions of \newprot{}'s security guarantees and the corresponding proofs. Our proofs use the notion of IND-sWKID-CPA security defined in \S{}3 of \cite{abdalla2007generalized}. They also depend on a property of the construction of \enc{} called \emph{history-independence}, which we explain below.

\subsection{Definition of History-Independence and Proof of Theorem \ref{thm:anonymity}}
Informally, a \enc{} construction is \textit{history-independent} if, for any fixed pattern $S$, the result of $\mathbf{KeyDer}$ to produce a key with pattern $S$, assuming that the starting key is either the master key or corresponds to a pattern that matches $S$, is distributed in exactly the same way regardless of the particular starting key used. The idea is that, given a key for a specified pattern $S$, one learns nothing about the sequence of $\mathbf{KeyDer}$ operations that produced the key.

We formally define history-independence below:
\begin{definition}[History-Independence]
A \enc{} construction is said to be \textit{history-independent} if, for every pattern $S$, and for any two well-formed keys $k_1$ (corresponding to pattern $P_1$) and $k_2$ (corresponding to pattern $P_2$) in the same \enc{} system such that $P_1$ matches $S$ and $P_2$ matches $S$, it holds that
$$
\{\mathbf{KeyDer}(k_1, S)\} = \{\mathbf{KeyDer}(k_2, S)\}
$$
where the distributions are over the randomness sampled internally by $\mathbf{KeyDer}$. If $k_1$ (respectively, $k_2$) is the master key, the pattern $P_1$ (respectively, $P_2$) is one where all slots are free.
\end{definition}

Below, we show that the construction of \enc{} presented in \appref{app:enc}, which \newprot{} uses, satisfies this property.

\begin{theorem}\label{thm:history_independence}
The construction of \enc{} presented in \appref{app:enc} is history-independent.
\end{theorem}
\begin{proof}[Proof of Theorem \ref{thm:history_independence}]
We will show that for any pattern $S$ and any well-formed key $k$ corresponding to a pattern $P$ that matches $S$, it holds that
\begin{multline*}
\{\mathbf{KeyDer}(k, S)\} = \\
\left\{\left(
g_2^\alpha \cdot \left(g_3 \cdot\hspace{-2ex}\prod_{(i, a_i) \in \fixed(S)} h_i^{a_i}\right)^r,
\ g^r,
\ \left\{(j, h_j^r)  \right\}_{j \in \free(S)} \right)\right\}_{r \sample \mathbb{Z}_p}
\end{multline*}
Because the formula on the right-hand side of the above equation only depends on $S$ and the public parameters (not the particular key $k$), this is sufficient to demonstrate history-independence of the \enc{} construction.

We handle the proof in two cases:

\parhead{Case 1} Suppose that $k$ is the master key. Then the above result is true by definition, according to the formula given for $\mathbf{KeyDer}$ in \appref{ssec:enc_construction}.

\parhead{Case 2} Suppose that $k$ is not the master key. Then, because $k$ is well-formed, we can write that
\begin{equation*}
k = \left(
g_2^\alpha \cdot \left(g_3 \cdot\hspace{-2ex}\prod_{(i, a_i) \in \fixed(P)} h_i^{a_i}\right)^{r_0},
\ g^{r_0},
\ \left\{(j, h_j^{r_0})  \right\}_{j \in \free(P)} \right)
\end{equation*}
for some fixed $r_0 \in \mathbb{Z}_p$. By applying the formula for $\mathbf{KeyDer}$ in \appref{ssec:enc_construction}, we can see that the key output by $\mathbf{KeyDer}$ has the form
\begin{equation*}
\left(
g_2^\alpha \cdot \left(g_3 \cdot\hspace{-2ex}\prod_{(i, a_i) \in \fixed(S)} h_i^{a_i}\right)^{r_0 + t},
\ g^{r_0 + t},
\ \left\{(j, h_j^{r_0 + t})  \right\}_{j \in \free(S)} \right)
\end{equation*}
for $t\sample\mathbb{Z}_p$. Because $r_0 + t$ is uniformly distributed in $\mathbb{Z}_p$, the output key has the desired distribution (take $r = r_0 + t$).
\end{proof}

Theorem \ref{thm:anonymity} follows directly from the fact that the \enc{} construction used in \newprot{} is history-independent: each ``signature'' in \enc{} is the same as a private key generated with a special slot filled in with the message being signed. Therefore, signatures inherit the history-independence of keys, resulting in the property in Theorem \ref{thm:anonymity}. With the proposed improvement to make signatures constant size (\secref{s:opt_anon_sig}), the signature consists of just the first two terms of the resulting private key, but it remains history-independent nonetheless.

For completeness, we prove Theorem \ref{thm:anonymity} below, using the same notation for signatures established in \appref{ssec:enc_construction}. The proof is very similar to the proof of Theorem \ref{thm:history_independence}
\begin{proof}[Proof of Theorem \ref{thm:anonymity}]
We will show that for any pattern $S$, key $k$ corresponding to pattern $S$, and message $m$, it holds that
\begin{multline*}
\{\mathbf{Sign}(k, m)\} =\\
\left\{\left(
g_2^\alpha \cdot \left(g_3 \cdot h_s^m \cdot \hspace{-2ex} \prod_{(i, a_i) \in \fixed(S)} h_i^{a_i}\right)^r,\ g^r
\right)\right\}_{r \sample \mathbb{Z}_p}
\end{multline*}
Because the right-hand side of the above equation depends only on $S$ and the public parameters (not the particular key $k$), this is sufficient to prove Theorem \ref{thm:anonymity} (that any two keys corresponding to $S$ produce the same signature distribution).

Observe that for a well-formed key $k$,
\begin{equation*}
k = \left(
g_2^\alpha \cdot \left(g_3 \cdot\hspace{-2ex}\prod_{(i, a_i) \in \fixed(S)} h_i^{a_i}\right)^{r_0},
\ g^{r_0},
\ \left\{(j, h_j^{r_0})  \right\}_{j \in \free(S)} \right)
\end{equation*}
for some fixed $r_0 \in \mathbb{Z}_p$. Applying the formula for $\mathbf{Sign}$ in \appref{ssec:enc_construction}, the signature has the form
\begin{equation*}
\left(
g_2^\alpha \cdot \left(g_3 \cdot h_s^m \cdot \hspace{-2ex} \prod_{(i, a_i) \in \fixed(S)} h_i^{a_i}\right)^{r_0+t},\ g^{r_0+t}
\right)
\end{equation*}
for $t\sample\mathbb{Z}_p$. Because $r_0 + t$ is uniformly distributed in $\mathbb{Z}_p$, the output signature has the desired distribution (take $r = r_0 + t$).
\end{proof}

Finally, we note that Theorem \ref{thm:anonymity} and its proof  easily extend to the ``generalized'' signatures ($\mathbf{GeneralizedSign}$) discussed in \appref{ssec:enc_construction}, where a key for pattern $P$ can generate a signature for \emph{another} pattern $S$ where $P$ matches $S$. In this case, the theorem guarantees that for any message $m$, pattern $S$, and two well-formed keys $k_1$ and $k_2$ with patterns $P_1$ and $P_2$ (in the same resource hierarchy), \textbf{if} $P_1$ matches $S$ and $P_2$ matches $S$, \textbf{then} signatures for pattern $S$ and message $m$ generated using $k_1$ are distributed identically to signatures for pattern $S$ and message $m$ generated using $k_2$, even if $P_1 \neq P_2$. The proof for this generalization of Theorem \ref{thm:anonymity} is very similar to the proofs of Theorem \ref{thm:history_independence} and Theorem \ref{thm:anonymity}.

\subsection{Proof of Theorem~\ref{thm:newprot_game}}

Below, we prove Theorem~\ref{thm:newprot_game} from \secref{s:encryption_guarantee}. Some intuition behind the proof is that the challenger in the game in Theorem~\ref{thm:newprot_game} (which is also the adversary in the IND-sWKID-CPA game~\cite{abdalla2007generalized}), maintains, for each \ent{}, the set of keys it has.
It \emph{lazily} requests these keys from the IND-sWKID-CPA challenger as \ents{} are compromised. Therefore, it maintains (1) a \emph{key set} for each \ent{}, storing the keys requested from the IND-sWKID-CPA challenger, and (2) a \emph{pattern set} for each \ent{}, storing patterns corresponding to additional keys that the \ent{} would have in the normal course of \newprot{}, but which have not been requested from the IND-sWKID-CPA challenger yet.
Requesting keys lazily is crucial because an uncompromised \ent{} in Theorem~\ref{thm:newprot_game} may possess a secret key capable of decrypting the challenge ciphertext.
\begin{proof}[Proof of Theorem \ref{thm:newprot_game}]
We show that, given an adversary $\mathcal{A}$ with non-negligible advantage in the game in Theorem~\ref{thm:newprot_game}, one can construct an algorithm $\mathcal{B}$ with non-negligible advantage in the IND-sWKID-CPA security game~\cite{abdalla2007generalized}. We denote as $\mathcal{C}$ the IND-sWKID-CPA security challenger. $\mathcal{B}$ maintains the following state: (2) a mapping from \ent{} name (in $\{0, 1\}^*$) to a key set for that \ent{}, and (3) a mapping from \ent{} name (in $\{0, 1\}^*$) to a pattern set for that \ent{}. These two maps are initialized as follows; each has a single entry for the name corresponding to the authority. The authority's key set is empty, and its pattern set contains one element, namely a pattern containing $\bot$ in all components (i.e., with all slots free).

$\mathcal{B}$ first runs the game with $\mathcal{A}$ as the challenger. $\mathcal{A}$ specifies the pair (URI, time) that it will attack at the beginning of the game. $\mathcal{B}$ parses (URI, time) into the pattern $S^*$ and gives it to $\mathcal{C}$. $\mathcal{C}$ generates the master key pair $(\mpk{},\msk{}) \leftarrow \mathsf{Setup}$ and gives $B$ the master public key $\mpk{}$. $\mathcal{B}$ forwards $\mpk{}$ to $\mathcal{A}$. For any of three queries from $\mathcal{A}$ in Phase 1, $\mathcal{B}$ processes it as following:
\begin{itemize}[leftmargin=*, noitemsep]
    \item \calA{} asks \calB{} to create \anent{}: \calB{} returns a fresh name in $\{0, 1\}^*$ corresponding to the new \ent{}. \calB{} creates mappings from this name to an empty set, for both the key set and pattern set, indicating that this new principal has not been delegated any keys.
    \item \calA{} asks \calB{} for the key set of a \ent{} $p$: \calB{} finds in its local state the key set and pattern set for $p$. For each pattern in $p$'s pattern set, it queries \calA{} for the corresponding \enc{} secret key. It adds each \enc{} secret key to $p$'s key set, and then replaces $p$'s pattern set in its local state with an empty set. Then it returns the keys in $p$'s key set to \calA{}. Note that \calB{} will not query \calC{} the secret key for a pattern that matches $S^*$, because, in the game in Theorem \ref{thm:newprot_game}, \calA{} is not allowed to request a key set containing a key whose URI and time match the challenge pair (URI, time). Also note that the keys given to \calA{} are distributed exactly as they would be in the \newprot{} protocol, because the underlying \enc{} scheme is assumed to be history-independent.
    \item \calA{} asks an \ent{} $p$ to make a delegation of \calA{}'s choice of another \ent{} $q$: \calB{} finds in its local state the key set and pattern set for $p$. \calB{} obtains the pattern corresponding to each key in $p$'s key set. Let $M$ be the set containing those patterns. \calB{} computes the set $N$, which is the union of $M$ and $p$'s pattern set. Based on the patterns in $N$, \calB{} computes the patterns corresponding to the keys that $p$ would generate and delegate to $q$. For each such key, \calB{} adds the corresponding pattern to $q$'s pattern set.
\end{itemize}
At the end of Phase 1, \calA{} outputs two equal-length challenge messages $m_0$ and $m_1$, and sends them to \calB{}. \calB{} then forwards $m_0$ and $m_1$ to \calC{}. \calC{} chooses a random bit $b$, and sends \calB{} the ciphertext of $m_b$. \calB{} then forwards the ciphertext to \calA{}.

In Phase 2, \calA{} makes additional queries as in Phase 1, and \calC{} can process them as before.

Finally, \calA{} will return the bit $b'$. \calB{} returns $b'$ to \calC{}. Because every response that \calB{} makes to \calA{} is distributed identically to the results of actually playing the game in Theorem \ref{thm:newprot_game}, \calA{} will guess $b' = b$ with non-negligible advantage. Thus, \calB{} wins the IND-sWKID-CPA game with non-negligible advantage.
\end{proof}

\subsection{Security Guarantee of \newprot{}'s Revocation}
\begin{theorem}
\label{thm:revoc_game}
Suppose revocation in \newprot{} is instantiated with a Selective-ID CPA-secure~\cite{canetti2003forward, boneh2004efficient, abdalla2007generalized}, history-independent \enc{} scheme. Then, there exists no probabilisitc polynomial-time adversary $\mathcal{A}$ who can win the following security game against a challenger $\mathcal{C}$ with more than negligible advantage:\\
\textbf{Initialization.} $\mathcal{A}$ selects a (URI, time) pair and a list \calL{} of revoked \leaves{} to attack.\\
\textbf{Setup.} $\mathcal{C}$ gives $\mathcal{A}$ the public parameters of the \newprot{} instance.\\
\textbf{Phase 1.} $\mathcal{A}$ can make three types of queries to $\mathcal{C}$.
\begin{enumerate}[topsep=0pt, noitemsep, wide, labelwidth=!, labelindent=0pt]
    \item $\mathcal{A}$ asks $\mathcal{C}$ to create \anent{}; $\mathcal{C}$ returns a name in $\{0, 1\}^*$, which $\mathcal{A}$ can use to refer to that \ent{} in future queries. A special name exists for the authority.
    \item $\mathcal{A}$ can ask $\mathcal{C}$ for the key set of any \ent{}; $\mathcal{C}$ gives $\mathcal{A}$ the keys that the \ent{} has. The only restriction is that, at the time this query is made, every key in the requested key set whose URI and time are prefixes of the challenge (URI, time) must have \leaves{} that are entirely in the challenge list \calL{}.
    \item $\mathcal{A}$ can ask $\mathcal{C}$ to make any \ent{} to make a delegation of $\mathcal{A}$'s choice to another \ent{}. The new \ent{} may have restricted pattern or fewer \leaves{}.
\end{enumerate}
\textbf{Challenge.} When $\mathcal{A}$ chooses to end Phase 1, it sends $\mathcal{C}$ two messages, $m_0$ and $m_1$, of the same length. Then $\mathcal{C}$ chooses a random bit $b \in \{0, 1\}$, encrypts $m_b$ under the challenge (URI, time) pair and list \calL{} of revoked \leaves{}, and gives $\mathcal{A}$ all of the ciphertexts.\\
\textbf{Phase 2.} $\mathcal{A}$ can make additional queries as in Phase 1.\\
\textbf{Guess.} $\mathcal{A}$ outputs $b' \in \{0, 1\}$, and wins the game if $b = b'$.
The advantage of an adversary $\mathcal{A}$ is $\left| \Pr[\mathcal{A} \text{ wins}] - \frac{1}{2} \right|$.
\end{theorem}
\begin{proof}
The proof of Theorem \ref{thm:revoc_game} is somewhat trickier than the proof of Theorem \ref{thm:newprot_game} because encrypted messages in \newprot{}, when revocation is used, consist of multiple \enc{} ciphertexts. Therefore, we use a hybrid argument. The hybrid game $\hyb^{(0)}$ is one that the adversary $\adv$ has no chance of winning. Hybrid $\hyb^{(n)}$ is a game that perfectly simulates \enc{}'s revocation protocol. We prove using security of the underlying \enc{} scheme that for all $i \in \{1, ..., n\}$, the difference between $\adv$'s advantage in hybrid game $\hyb^{(i - 1)}$ and hybrid game $\hyb^{(i)}$ is negligible.

The number $n$, which controls the number of hybrids, is defined to be the size of the subset cover, which depends on the list \calL{} that $\adv$ declares in the Initialization phase. Note that this is the same as the number of \enc{} ciphertexts in an encrypted message with revocation list \calL{}.

We define the hybrid game $\hyb^{(i)}$, for $i \in \{0, ..., n\}$, as identical to the game given in Theorem \ref{thm:revoc_game}, with the following difference. The \newprot{} ciphertext returned to $\adv$ in the Challenge phase consists of $n$ \enc{} ciphertexts; $\hyb^{(i)}$ generates the first $i$ ciphertexts correctly, and then replaces each of the remaining $n - i$ ciphertexts with an encryption of $0$ under the same pattern. Observe that $\adv$ has no chance of winning $\hyb^{(0)}$, because the challenge ciphertext is chosen independently of the bit $b$ chosen by the challenger. Also observe that $\hyb^{(n)}$ is identical to the game described in Theorem \ref{thm:revoc_game}.

All that remains to prove is that the difference between $\adv$'s advantage in game $\hyb^{(i-1)}$ and game $\hyb^{(i)}$ is negligible for all $i \in \{1, \ldots, n\}$. We do so via a reduction: given a probabilistic polynomial-time adversary $\adv$ whose difference in advantage is non-negligible, we construct a probabilistic polynomial-time adversary \calB{} that wins the IND-sWKID-CPA game with non-negligible probability. In our reduction, \calB{} acts as the challenger in the hybrid game; we denote the challenger in the IND-sWKID-CPA game as \calC{}.

In the Initialization phase, $\mathcal{A}$ specifies the pair (URI, time) and revocation list \calL{} that it will attack. Then, $\mathcal{B}$ computes the subset cover over all \leaves{} not in \calL{}, and selects $\mathsf{ID}_i$, the ID of $i$th subset in the subset cover. $\mathcal{B}$ parses (URI, time) and $\mathsf{ID}_i$ into the pattern $S^*$ and gives it to $\mathcal{C}$. $\mathcal{C}$ generates the master key pair $(\mpk{},\msk{}) \leftarrow \mathsf{Setup}$ and gives $B$ the master public key $\mpk{}$. $\mathcal{B}$ forwards $\mpk{}$ to $\mathcal{A}$.

For any of three queries from $\mathcal{A}$ in Phase 1, $\mathcal{B}$ processes it as it does in the Proof of Theorem \ref{thm:newprot_game}. The only difference is that \calA{} also specifies which \leaves{} are included in each delegation (query \#3), and \calB{} takes this into account when generating the pattern set $N$. Note that, because of the ``subset cover'' property used by tree-based broadcast encryption, the \enc{} patterns used to query \calA{} for \enc{} keys will never match $S^*$.

At the end of Phase 1, \calA{} outputs two equal-length challenge messages $m_0$ and $m_1$, and sends them to \calB{}. \calB{} then chooses a random bit $b^*$. \calB{} computes the \newprot{} ciphertext, which consists of $n$ \enc{} ciphertexts as follows. To compute the $j$th \enc{} ciphertext, where $1 \leq j \leq i - 1$, it encrypts $0$ with the pattern corresponding to the challenge (URI, time) and $\mathsf{ID}_j$.
To compute $j$th \enc{} ciphertext, where $i + 1 \leq j \leq n$, it encrypts $m_{b^*}$ with the pattern corresponding to the challenge (URI, time) and $\mathsf{ID}_j$. To compute the $i$th \enc{} ciphertext, it forwards $0$ and $m_{b^*}$ to \calC{}.
\calC{} chooses a random bit $b$, and sends \calB{} either an encryption of $0$ or an encryption of $m_{b^*}$ depending on $b$.the ciphertext of $m_b$. \calB{} uses this as the $i$th ciphertext. It assembles the $n$ \enc{} ciphertexts, computed as above, into a \newprot{} ciphertext, and forwards them to \calA{}. Note that if $b = 0$, then \calB{} played game $\hyb^{(i - 1)}$, and if $b = 1$, then $\calB{}$ played game $\hyb^{(i)}$.

In Phase 2, \calA{} makes additional queries as in Phase 1, and \calC{} can process them as before.

Finally, \calA{} will return the bit $b'$. \calB{} checks if $b' = b^*$. If $b' = b^**$, then \calB{} guesses that $b = 1$; otherwise, it guesses that $b = 0$. It sends its guess to \calC{}. Because \calA{} is assumed to have a non-negligible difference in advantage between $\hyb^{(i - 1)}$ and $\hyb^{(i)}$, \calB{}'s advantage in the IND-sWKID-CPA game is non-negligible.
\end{proof}

\subsection{Adaptive Security}
A natural question is how to achieve adaptive security. As has been observed for IBE~\cite{boneh2004efficient}, HIBE~\cite{boneh2005hierarchical}, and WKD-IBE~\cite{abdalla2007generalized}, hashing each component of the ID results in adaptive security, but with a loss of security exponential in the size of the hash.
However, if the hash function is modeled as a random oracle, and the number of slots in \enc{} is logarithmic in the security parameter, then the loss in security is polynomial~\cite{abdalla2007generalized} (assuming that the number of slots $\ell$ is logarithmic in the security parameter).
Given that we use a hash function to map URI/time to a pattern (\secref{ssec:uri_time}), this analysis applies to \newprot{}.

%% file: f_extensions.tex
\section{Extensions}\label{sec:extensions}

We present two extensions to \newprot{}'s core encryption protocol (\secref{sec:encryption}): (1) generalized subtrees with wildcards in the middle of a URI, and (2) forward secrecy.

\subsection{Beyond Simple Hierarchies}\label{ssec:plus_wildcard}

Thus far, we have considered only the \texttt{*} wildcard at the end of a URI. With \enc{}, we can also place a \texttt{+} wildcard in the middle of a URI, allowing a single component of the URI to remain unspecified. For example, the URI \path{a/+/b} matches all URIs of length 3 where the first component is \texttt{a} and the third component is \texttt{b}; the second component could be anything. To implement the \texttt{+} wildcard, we fill in the components corresponding to \texttt{+} with $\bot$.

The \texttt{+} wildcard is useful in real applications. For example, in the ``smart buildings'' setting, one could imagine a resource hierarchy of the form \path{buildingA/floor2/room/sensor_id/reading_type}, where \texttt{reading\_type} could be either \texttt{temp} or \texttt{hum}. The \texttt{+} wildcard allows one to delegate permission to see only the temperature readings in a building, by granting permission on the URI \path{buildingA/+/+/+/temp}. It is also useful for the time hierarchy. An organization may want to give an employee access to a resource from 8 AM to 5 PM every day, which can be accomplished by using the \texttt{+} wildcard for the slots corresponding to the year, month, and day. See \figref{fig:slots_plus}.

\begin{figure}
    \centering
    \includegraphics[width=\linewidth]{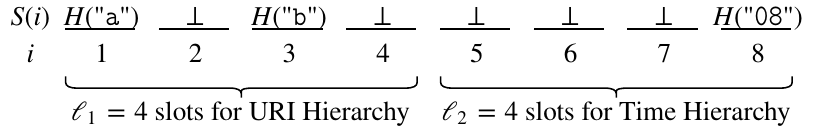}
    \caption{Pattern $S$ for a private key granting access to \texttt{a/+/b/*} at 8 AM every day. The figure uses 8 slots for space reasons; \newprot{} is meant to be used with more slots (e.g., 20).}
    \label{fig:slots_plus}
    \vspace{-3ex}
\end{figure}

\subsection{Forward Secrecy}\label{ssec:forward_security}

Forward secrecy is the property that if a subscriber's decryption key is compromised, the attacker should only be able to decrypt \emph{new} messages visible to the subscriber, not old messages sent before the key was compromised.

Forward secrecy using HIBE has been studied in~\cite{canetti2003forward}. We can apply the same idea to our construction via a straightforward extension to our mechanism for expiry. In our construction of expiry, each subscriber has a collection of keys for each URI or URI prefix it can access that give it access over a time range $[t_1, t_2]$, which can be qualified to any smaller time range $[t_3, t_4]$ where $t_1 \leq t_3 \leq t_4 \leq t_2$. To achieve forward security, each subscriber qualifies the keys for each URI, at each unit of time, to only be valid starting at the \emph{current time} until the same expiry time, and then discards the old keys. This guarantees that, if a key is stolen, it cannot be used to decrypt messages published before the current time.

%% file: main.bbl
\begin{thebibliography}{10}

\bibitem{abdalla2007generalized}
M.~Abdalla, E.~Kiltz, and G.~Neven.
\newblock Generalized key delegation for hierarchical identity-based
  encryption.
\newblock Cryptology ePrint Archive, Report 2007/221.

\bibitem{akl1983cryptographic}
S.~G. Akl and P.~D. Taylor.
\newblock Cryptographic solution to a problem of access control in a hierarchy.
\newblock {\em TOCS}, 1983.

\bibitem{andersen2016system}
M.~P Andersen, G.~Fierro, and D.~E. Culler.
\newblock System design for a synergistic, low power mote/{BLE} embedded
  platform.
\newblock In {\em IPSN}, 2016.

\bibitem{andersen2017hamilton}
M.~P Andersen, H.-S. Kim, and D.~E. Culler.
\newblock {Hamilton} - a cost-effective, low power networked sensor for indoor
  environment monitoring.
\newblock In {\em BuildSys}, 2017.

\bibitem{andersen2017democratizing}
M.~P Andersen, J.~Kolb, K.~Chen, D.~E. Culler, and R.~Katz.
\newblock Democratizing authority in the built environment.
\newblock In {\em BuildSys}, 2017.

\bibitem{andersen2019wave}
M.~P Andersen, S.~Kumar, M.~AbdelBaky, G.~Fierro, J.~Kolb, H.-S. Kim, D.~E.
  Culler, and R.~A. Popa.
\newblock {WAVE}: A decentralized authorization framework with transitive
  delegation.
\newblock In {\em USENIX Security}, 2019.

\bibitem{arjunan2012sensoract}
P.~Arjunan, N.~Batra, H.~Choi, A.~Singh, P.~Singh, and M.~B. Srivastava.
\newblock {SensorAct}: A privacy and security aware federated middleware for
  building management.
\newblock In {\em BuildSys}, 2012.

\bibitem{atallah2009dynamic}
M.~J. Atallah, M.~Blanton, N.~Fazio, and K.~B. Frikken.
\newblock Dynamic and efficient key management for access hierarchies.
\newblock In {\em TISSEC}, 2009.

\bibitem{atallah2007incorporating}
M.~J. Atallah, M.~Blanton, and K.~B. Frikken.
\newblock Incorporating temporal capabilities in existing key management
  schemes.
\newblock In {\em ESORICS}, 2007.

\bibitem{attrapadung2009conjunctive}
N.~Attrapadung and H.~Imai.
\newblock Conjunctive broadcast and attribute-based encryption.
\newblock In {\em ICPBC}, 2009.

\bibitem{barbulescu2017updating}
R.~Barbulescu and S.~Duquesne.
\newblock Updating key size estimations for pairings.
\newblock Cryptology ePrint Archive, Report 2017/334.

\bibitem{belguith2018secure}
S.~Belguith, S.~Cui, M.~R. Asghar, and G.~Russello.
\newblock Secure publish and subscribe systems with efficient revocation.
\newblock In {\em SAC}, 2018.

\bibitem{bethencourt2007ciphertext}
J.~Bethencourt, A.~Sahai, and B.~Waters.
\newblock Ciphertext-policy attribute-based encryption.
\newblock In {\em S\&P}, 2007.

\bibitem{birgisson2014macaroons}
A.~Birgisson, J.~G. Politz, Ú. Erlingsson, A.~Taly, M.~Vrable, and
  M.~Lentczner.
\newblock Macaroons: Cookies with contextual caveats for decentralized
  authorization in the cloud.
\newblock In {\em NDSS}, 2014.

\bibitem{blaze1998divertible}
M.~Blaze, G.~Bleumer, and M.~Strauss.
\newblock Divertible protocols and atomic proxy cryptography.
\newblock In {\em EUROCRYPT}, 1998.

\bibitem{boldyreva2008identity}
A.~Boldyreva, V.~Goyal, and V.~Kumar.
\newblock Identity-based encryption with efficient revocation.
\newblock In {\em CCS}, 2008.

\bibitem{boneh2004efficient}
D.~Boneh and X.~Boyen.
\newblock Efficient selective-{ID} secure identity-based encryption without
  random oracles.
\newblock In {\em EUROCRYPT}, 2004.

\bibitem{boneh2005hierarchical}
D.~Boneh, X.~Boyen, and E.-J. Goh.
\newblock Hierarchical identity based encryption with constant size ciphertext.
\newblock In {\em EUROCRYPT and Cryptology ePrint Archive}, 2005.

\bibitem{boneh2001identity}
D.~Boneh and M.~Franklin.
\newblock Identity-based encryption from the {Weil} pairing.
\newblock In {\em CRYPTO}, 2001.

\bibitem{boneh2005collusion}
D.~Boneh, C.~Gentry, and B.~Waters.
\newblock Collusion resistant broadcast encryption with short ciphertexts and
  private keys.
\newblock In {\em CRYPTO}, 2005.

\bibitem{boneh2006fully}
D.~Boneh and B.~Waters.
\newblock A fully collusion resistant broadcast, trace, and revoke system.
\newblock In {\em CCS}, 2006.

\bibitem{boneh2014low}
D.~Boneh, B.~Waters, and M.~Zhandry.
\newblock Low overhead broadcast encryption from multilinear maps.
\newblock In {\em CRYPTO}, 2014.

\bibitem{boneh2017multiparty}
D.~Boneh and M.~Zhandry.
\newblock Multiparty key exchange, efficient traitor tracing, and more from
  indistinguishability obfuscation.
\newblock {\em Algorithmica}, 2017.

\bibitem{borcea2017picador}
C.~Borcea, A.~B.~D. Gupta, Y.~Polyakov, K.~Rohloff, and G.~Ryan.
\newblock {PICADOR}: End-to-end encrypted publish-subscribe information
  distribution with proxy re-encryption.
\newblock {\em FGCS}, 2017.

\bibitem{bowe2018bls12}
S.~Bowe.
\newblock {BLS12-381}: New {zk-SNARK} elliptic curve construction, 2018.
\newblock \url{https://z.cash/blog/new-snark-curve/}.

\bibitem{winter2012rpl}
A.~Brandt, J.~Hui, R.~Kelsey, P.~Levis, K.~Pister, R.~Struik, J.~P. Vasseur,
  and R.~Alexander.
\newblock {RPL}: {IPv6} routing protocol for low-power and lossy networks.
\newblock {RFC}, RFC Editor, 2012.

\bibitem{brunelli2014povomon}
D.~Brunelli, I.~Minakov, R.~Passerone, and M.~Rossi.
\newblock {POVOMON}: An ad-hoc wireless sensor network for indoor environmental
  monitoring.
\newblock In {\em EESMS}, 2014.

\bibitem{bw2}
bw2.
\newblock \url{https://github.com/immesys/bw2}.

\bibitem{campbell2017hail}
B.~Campbell.
\newblock Introducing {Hail}, 2017.
\newblock \url{https://www.tockos.org/blog/2017/introducing-hail/}.

\bibitem{canetti2003forward}
R.~Canetti, S.~Halevi, and J.~Katz.
\newblock A forward-secure public-key encryption scheme.
\newblock In {\em EUROCRYPT}, 2003.

\bibitem{cheng2016talek}
R.~Cheng, W.~Scott, B.~Parno, I.~Zhang, A.~Krishnamurthy, and T.~Anderson.
\newblock Talek: A private publish-subscribe protocol.
\newblock Technical report, University of Washington CSE, 2016.

\bibitem{cheon2006security}
J.~H. Cheon.
\newblock Security analysis of the strong {Diffie}-{Hellman} problem.
\newblock In {\em EUROCRYPT}, 2006.

\bibitem{cisco2014internet}
Cisco.
\newblock The {Internet} of things reference model.
\newblock Technical report, Cisco, 2014.

\bibitem{clarke2001certificate}
D.~Clarke, J.-E. Elien, C.~Ellison, M.~Fredette, A.~Morcos, and R.~L. Rivest.
\newblock Certificate chain discovery in {SPKI/SDSI}.
\newblock {\em Journal of Computer Security}, 2001.

\bibitem{corrigan2015riposte}
H.~Corrigan-Gibbs, D.~Boneh, and D.~Mazi\`{e}res.
\newblock Riposte: An anonymous messaging system handling millions of users.
\newblock In {\em S\&P}, 2015.

\bibitem{crampton2015cryptographic}
J.~Crampton, N.~Farley, G.~Gutin, M.~Jones, and B.~Poettering.
\newblock Cryptographic enforcement of information flow policies without public
  information.
\newblock In {\em ACNS}, 2015.

\bibitem{crampton2006key}
J.~Crampton, K.~Martin, and P.~Wild.
\newblock On key assignment for hierarchical access control.
\newblock In {\em CSFW}, 2006.

\bibitem{haggerty2010smap}
S.~Dawson-Haggerty, X.~Jiang, G.~Tolle, J.~Ortiz, and D.~E. Culler.
\newblock {sMAP}: A simple measurement and actuation profile for physical
  information.
\newblock In {\em SenSys}, 2010.

\bibitem{haggerty2013boss}
S.~Dawson-Haggerty, A.~Krioukov, J.~Taneja, S.~Karandikar, G.~Fierro,
  N.~Kitaev, and D.~E. Culler.
\newblock {BOSS}: Building operating system services.
\newblock In {\em NSDI}, 2013.

\bibitem{dodis2002public}
Y.~Dodis and N.~Fazio.
\newblock Public key broadcast encryption for stateless receivers.
\newblock In {\em DRM}, 2002.

\bibitem{dutta2007procrastination}
P.~Dutta, D.~E. Culler, and S.~Shenker.
\newblock Procrastination might lead to a longer and more useful life.
\newblock In {\em HotNets}, 2007.

\bibitem{egorov2017nucypher}
M.~Egorov and M.~Wilkison.
\newblock {NuCypher KMS}: Decentralized key management system.
\newblock {\em CoRR}, 2017.

\bibitem{samr21e18a}
DigiKey Electronics.
\newblock {ATSAMR21E18A-MU} microchip technology.
\newblock Feb. 8, 2019.

\bibitem{feldmeier2009personalized}
M.~C. Feldmeier.
\newblock {\em Personalized Building Comfort Control}.
\newblock PhD thesis, MIT, 2009.

\bibitem{fierro2015xbos}
G.~Fierro and D.~E. Culler.
\newblock {XBOS}: An extensible building operating system.
\newblock Technical report, EECS Department, University of California,
  Berkeley, 2015.

\bibitem{filecoin}
Filecoin.
\newblock \url{https://filecoin.io}.
\newblock Jan. 19, 2018.

\bibitem{frosch2016how}
T.~Frosch, C.~Mainka, C.~Bader, F.~Bergsma, J.~Schwenk, and T.~Holz.
\newblock How secure is {TextSecure}?
\newblock In {\em EuroS\&P}, 2016.

\bibitem{gentry2009hierarchical}
C.~Gentry and S.~Halevi.
\newblock Hierarchical identity based encryption with polynomially many levels.
\newblock In {\em TCC}, 2009.

\bibitem{gentry2002hierarchical}
C.~Gentry and A.~Silverberg.
\newblock Hierarchical {ID}-based cryptography.
\newblock In {\em ASIACRYPT}, 2002.

\bibitem{goh2003sirius}
E.-J. Goh, H.~Shacham, N.~Modadugu, and D.~Boneh.
\newblock {SiRiUS}: Securing remote untrusted storage.
\newblock In {\em NDSS}, 2003.

\bibitem{goyal2006attribute}
V.~Goyal, O.~Pandey, A.~Sahai, and B.~Waters.
\newblock Attribute-based encryption for fine-grained access control of
  encrypted data.
\newblock In {\em CCS}, 2006.

\bibitem{hamiltoniot}
Hamilton {IoT}.
\newblock \url{https://hamiltoniot.com/}.

\bibitem{hu2005efficient}
Y.-C. Hu, M.~Jakobsson, and A.~Perrig.
\newblock Efficient constructions for one-way hash chains.
\newblock In {\em ACNS}, 2005.

\bibitem{huang2004new}
H.-F. Huang and C.-C. Chang.
\newblock A new cryptographic key assignment scheme with time-constraint access
  control in a hierarchy.
\newblock {\em Computer Standards \& Interfaces}, 2004.

\bibitem{hviid2018activity}
J.~Hviid and M.~B. Kjaergaard.
\newblock Activity-tracking service for building operating systems.
\newblock In {\em PerCom}, 2018.

\bibitem{imix}
{imix}: Low-power {IoT} research platform, 2017.
\newblock \url{https://github.com/helena-project/imix}.

\bibitem{kallahalla2003plutus}
M.~Kallahalla, E.~Riedel, R.~Swaminathan, Q.~Wang, and K.~Fu.
\newblock Plutus: Scalable secure file sharing on untrusted storage.
\newblock In {\em FAST}, 2003.

\bibitem{kawahara2016barreto}
Y.~Kawahara, T.~Kobayashi, M.~Scott, and A.~Kato.
\newblock {Barreto}-{Naehrig} curves.
\newblock Technical report, Internet-Draft draft-kasamatsu-bncurves-02.
  Internet Engineering Task Force., 2016.
\newblock
  \url{https://datatracker.ietf.org/doc/html/draft-kasamatsu-bncurves-02}.

\bibitem{kim2018system}
H.-S. Kim, M.~P Andersen, K.~Chen, S.~Kumar, W.~J. Zhao, K.~Ma, and D.~E.
  Culler.
\newblock System architecture directions for post-{SoC}/32-bit networked
  sensors.
\newblock In {\em SenSys}, 2018.

\bibitem{kim2016extended}
T.~Kim and R.~Barbulescu.
\newblock Extended tower number field sieve: A new complexity for the medium
  prime case.
\newblock In {\em CRYPTO}, 2016.

\bibitem{krioukov2012building}
A.~Krioukov, G.~Fierro, N.~Kitaev, and D.~E. Culler.
\newblock Building application stack ({BAS}).
\newblock In {\em BuildSys}, 2012.

\bibitem{lewko2010revocation}
A.~Lewko, A.~Sahai, and B.~Waters.
\newblock Revocation systems with very small private keys.
\newblock In {\em S\&P}, 2010.

\bibitem{li2014energy}
C.~Li, Z.~Li, M.~Li, F.~Meggers, A.~Schlueter, and H.~B. Lim.
\newblock Energy efficient {HVAC} system with distributed sensing and control.
\newblock In {\em ICDCS}, 2014.

\bibitem{libert2012group}
B.~Libert, T.~Peters, and M.~Yung.
\newblock Group signatures with almost-for-free revocation.
\newblock In {\em CRYPTO}, 2012.

\bibitem{libert2012scalable}
B.~Libert, T.~Peters, and M.~Yung.
\newblock Scalable group signatures with revocation.
\newblock In {\em EUROCRYPT}, 2012.

\bibitem{libert2009adaptive}
B.~Libert and D.~Vergnaud.
\newblock Adaptive-{ID} secure revocable identity-based encryption.
\newblock In {\em CT-RSA}, 2009.

\bibitem{liu2016compact}
W.~Liu, J.~Liu, Q.~Wu, B.~Qin, D.~Naccache, and H.~Ferradi.
\newblock Compact {CCA2}-secure hierarchical identity-based broadcast
  encryption for fuzzy-entity data sharing.
\newblock Cryptology ePrint Archive, Report 2016/634.

\bibitem{malkhi2000secure}
D.~Malkhi, M.~Merritt, and O.~Rodeh.
\newblock Secure reliable multicast protocols in a {WAN}.
\newblock {\em Dist. Computing}, 2000.

\bibitem{malkhi1997high}
D.~Malkhi and M.~Reiter.
\newblock A high-throughput secure reliable multicast protocol.
\newblock {\em Computer Security}, 1997.

\bibitem{mehanovic2018brume}
A.~Mehanovic, T.~H. Rasmussen, and M.~B. Kjærgaard.
\newblock Brume - a horizontally scalable and fault tolerant building operating
  system.
\newblock In {\em IoTDI}, 2018.

\bibitem{merkle1989certified}
R.~C. Merkle.
\newblock A certified digital signature.
\newblock In {\em ASIACRYPT}, 1989.

\bibitem{naor2001revocation}
D.~Naor, M.~Naor, and J.~Lotspiech.
\newblock Revocation and tracing schemes for stateless receivers.
\newblock In {\em CRYPTO}, 2001.

\bibitem{neto2016aot}
A.~L.~M. Neto, A.~L.~F. Souza, I.~Cunha, M.~Nogueira, I.~O. Nunes, L.~Cotta,
  N.~Gentille, A.~A.~F. Loureiro, D.~F. Aranha, H.~K. Patil, and L.~B.
  Oliveira.
\newblock {AoT}: Authentication and access control for the entire {IoT} device
  life-cycle.
\newblock In {\em SenSys}, 2016.

\bibitem{particlemesh}
{Particle Mesh}.
\newblock \url{https://www.particle.io/mesh}.
\newblock Feb. 2, 2019.

\bibitem{perrig2001spins}
A.~Perrig, R.~Szewczyk, V.~Wen, D.~E. Culler, and J.~D. Tygar.
\newblock {SPINS}: Security protocols for sensor networks.
\newblock In {\em MobiCom}, 2001.

\bibitem{seo2013efficient}
J.~H. Seo and K.~Emura.
\newblock Efficient delegation of key generation and revocation functionalities
  in identity-based encryption.
\newblock In {\em CT-RSA}, 2013.

\bibitem{seo2013revocable}
J.~H Seo and K.~Emura.
\newblock Revocable identity-based encryption revisited: Security model and
  construction.
\newblock In {\em PKC}, 2013.

\bibitem{seo2015revocable}
J.~H. Seo and K.~Emura.
\newblock Revocable hierarchical identity-based encryption: History-free
  update, security against insiders, and short ciphertexts.
\newblock In {\em CT-RSA}, 2015.

\bibitem{shafagh2018droplet}
H.~Shafagh, L.~Burkhalter, S.~Duquennoy, A.~Hithnawi, and S.~Ratnasamy.
\newblock Droplet: Decentralized authorization for {IoT} data streams.
\newblock {\em CoRR}, 2018.

\bibitem{shafagh2017secure}
H.~Shafagh, A.~Hithnawi, L.~Burkhalter, P.~Fischli, and S.~Duquennoy.
\newblock Secure sharing of partially homomorphic encrypted {IoT} data.
\newblock In {\em SenSys}, 2017.

\bibitem{solace}
Solace cloud.
\newblock \url{https://solace.com}.
\newblock Jan. 17, 2018.

\bibitem{taly2016distributed}
A.~Taly and A.~Shankar.
\newblock Distributed authorization in {Vanadium}.
\newblock In {\em FOSAD VIII}, 2016.

\bibitem{tariq2014securing}
M.~A. Tariq, B.~Koldehofe, and K.~Rothermel.
\newblock Securing broker-less publish/subscribe systems using identity-based
  encryption.
\newblock {\em TPDS}, 2014.

\bibitem{tron2016smash}
V.~Tron, A.~Fischer, and N.~Johnson.
\newblock Smash-proof: Auditable storage for {Swarm} secured by masked audit
  secret hash.
\newblock Technical report, Ethersphere, 2016.

\bibitem{tzeng2002time}
W.-G. Tzeng.
\newblock A time-bound cryptographic key assignment scheme for access control
  in a hierarchy.
\newblock {\em TKDE}, 2002.

\bibitem{van2015vuvuzela}
J.~van~den Hooff, D.~Lazar, M.~Zaharia, and N.~Zeldovich.
\newblock Vuvuzela: Scalable private messaging resistant to traffic analysis.
\newblock In {\em SOSP}, 2015.

\bibitem{volttron}
{VOLTTRON}.
\newblock \url{https://volttron.org/}.
\newblock Jan. 23, 2019.

\bibitem{wang2016sieve}
F.~Wang, J.~Mickens, N.~Zeldovich, and V.~Vaikuntanathan.
\newblock Sieve: Cryptographically enforced access control for user data in
  untrusted clouds.
\newblock {\em NSDI}, 2016.

\bibitem{wang2010hierarchical}
G.~Wang, Q.~Liu, and J.~Wu.
\newblock Hierarchical attribute-based encryption for fine-grained access
  control in cloud storage services.
\newblock In {\em CCS}, 2010.

\bibitem{wang2011hierarchical}
G.~Wang, Q.~Liu, J.~Wu, and M.~Guo.
\newblock Hierarchical attribute-based encryption and scalable user revocation
  for sharing data in cloud servers.
\newblock {\em Computers \& Security}, 2011.

\bibitem{watanabe2017new}
Y.~Watanabe, K.~Emura, and J.~H. Seo.
\newblock New revocable {IBE} in prime-order groups: Adaptively secure,
  decryption key exposure resistant, and with short public parameters.
\newblock In {\em CT-RSA}, 2017.

\bibitem{wu2016privacy}
D.~J. Wu, A.~Taly, A.~Shankar, and D.~Boneh.
\newblock Privacy, discovery, and authentication for the {Internet} of things.
\newblock In {\em ESORICS}, 2016.

\bibitem{yao2004id}
D.~Yao, N.~Fazio, Y.~Dodis, and A.~Lysyanskaya.
\newblock {ID}-based encryption for complex hierarchies with applications to
  forward security and broadcast encryption.
\newblock In {\em CCS}, 2004.

\bibitem{ye2002energy}
W.~Ye, J.~Heidemann, and D.~Estrin.
\newblock An energy-efficient {MAC} protocol for wireless sensor networks.
\newblock In {\em INFOCOM}, 2002.

\bibitem{yu2010achieving}
S.~Yu, C.~Wang, K.~Ren, and W.~Lou.
\newblock Achieving secure, scalable, and fine-grained data access control in
  cloud computing.
\newblock In {\em INFOCOM}, 2010.

\bibitem{zachariah2015internet}
T.~Zachariah, N.~Klugman, B.~Campbell, J.~Adkins, N.~Jackson, and P.~Dutta.
\newblock The {Internet} of things has a gateway problem.
\newblock In {\em HotMobile}, 2015.

\bibitem{zigbeegateway}
Zigbee gateway.
\newblock \url{https://www.zigbee.org/zigbee-for-developers/zigbee-gateway/}.
\newblock Feb. 13, 2019.

\end{thebibliography}
